\documentclass[11pt]{article}
\usepackage[a4paper,left=15mm,right=15mm,top=15mm,bottom=15mm]{geometry}

\usepackage{authblk}
\usepackage[T1]{fontenc}
\usepackage{alemfonts}
\usepackage[english]{babel}
\usepackage{amsmath}
\usepackage{mathtools}
\usepackage{amsfonts}
\usepackage{amssymb}
\usepackage{xcolor}
\usepackage{mathrsfs}
\usepackage{amsthm}
\usepackage{float}
\usepackage{graphicx}
\usepackage{hyperref}
\usepackage[parfill]{parskip}
\usepackage{bm}
\usepackage{comment}

\newcommand{\pder}[2]{\ensuremath{\frac{\partial #1}{\partial #2}}}

\numberwithin{equation}{section}

\long\def \beq#1\eeq {\begin{equation} #1 \end{equation}}
\long\def \beaq#1\eeaq {\begin{equation}\begin{aligned} #1 \end{aligned}\end{equation}}
\long\def \bes#1\ees {\begin{equation}\begin{split} #1 \end{split} \end{equation}}
\long\def \bea#1\eea {\begin{eqnarray} #1 \end{eqnarray}}
\long\def \bse[#1]#2\ese {\begin{subequations}\label{#1}\begin{align} #2 \end{align}\end{subequations}}

\newcommand{\mv}[1]{\langle #1\rangle}

\newcommand{\sign}{\operatorname{sign}}

\long\def\dm[#1]{\!\operatorname{d\mu}\left(#1\right)}

\newcommand{\RS}{{\rm \scriptscriptstyle RS}}
\newcommand{\RSB}{{\rm \scriptscriptstyle 1RSB}}
\newcommand{\inter}{{\rm \scriptscriptstyle inter}}
\newcommand{\intra}{{\rm \scriptscriptstyle intra}}
\theoremstyle{plain}
\newtheorem{Remark}{Remark}
\newtheorem{Theorem}{Theorem}
\newtheorem{Lemma}{Lemma}
\newtheorem{Proposition}{Proposition}
\newtheorem{Corollary}{Corollary}
\newtheorem{Definition}{Definition}

\graphicspath{{./Figures/}}

\newcommand{\sums}{\sum_{ \boldsymbol \sigma}}
\newcommand{\sump}{\sum_{p=1}^{K}}
\newcommand{\sumpp}{\sum_{p=1}^{K}}
\newcommand{\sumpnew}{\sum_{p=1}^{K-1}}

\newcommand{\suma}{\sum_{a=1}^2}

\newcommand{\si}{\sigma_i}
\newcommand{\sj}{\sigma_j}
\newcommand{\dt}{\frac{\partial}{\partial t}}

\newcommand{\qb}{\bar{q}}
\newcommand{\mb}{\bar{m}}

\title{Pattern recognition in Deep Boltzmann machines}

\author[1,2]{E. Agliari}
\author[2,3]{L. Albanese}
\author[2,3]{F. Alemanno}
\author[2,3,4]{A. Fachechi}

\affil[1]{Dipartimento di Matematica {\em Guido Castelnuovo}, Sapienza Universit\`a  di Roma, Rome, Italy}
\affil[2]{Istituto Nazionale d'Alta Matematica {\em Francesco Severi},  Rome, Italy}
\affil[3]{Dipartimento di Matematica e Fisica {\em Ennio De Giorgi},  Universit\`a  del Salento,  Lecce, Italy}
\affil[4]{Istituto Nazionale di Fisica Nucleare, Campus Ecotekne, Lecce,Italy}

\begin{document}

\setlength{\columnsep}{20pt}

\addtocontents{toc}{~\hfill\textbf{Page}\par}
{\let\newpage\relax\maketitle}



\begin{abstract}
We consider a multi-layer Sherrington-Kirkpatrick spin-glass as a model for deep restricted Boltzmann machines and we solve for its quenched free energy, in the thermodynamic limit and allowing for a first step of replica symmetry breaking. This result is accomplished rigorously exploiting interpolating techniques and recovering the expression already known for the replica-symmetry case. Further, we drop the restriction constraint by introducing intra-layer connections among spins and we show that the resulting system can be mapped into a modular Hopfield network, which is also addressed rigorously via interpolating techniques up to the first step of replica symmetry breaking.
\end{abstract}

\tableofcontents
\section{Introduction}
Restricted Boltzmann machines (RBMs) constitute a popular model in machine learning (see e.g., \cite{lenkanature,app1,app2}) due to a relatively easy architecture (a RBM consists of one visible layer and one hidden layer with interlayer interactions only) and to its analogies with bipartite spin-glass models which allow for theoretical investigations and a mathematical control \cite{Ackley,BarraEquivalenceRBMeAHN,Barra-RBMsPriors1,Barra-RBMsPriors2,MarulloAgliari}.
A more challenging version of the RBM is given by the Deep Boltzmann machine (DBM), which comprises a set of hidden layers, see Fig.~\ref{fig:schema1}
In the statistical-mechanics of disordered-system jargon, this can be considered as a multi-layer spin-glass where the set of spins is arranged into a geometry made of consecutive layers and only interactions among spins belonging to adjacent layers are allowed. Also motivated by the impressive successes obtained in artificial intelligence via deep learning methods (which, beyond DBMs include a number of other different neural networks), this kind of structures have recently attracted a wide interest (see e.g., \cite{ZiqquratBarra,Bates,alberici1,alberici2,alberici3,Genovese2}).
\newline
Here, we consider a multi-layer spin-glass, specifically, a multi-layer Sherrington-Kirkpatrick (MSK) model and we solve for its free-energy via interpolating techniques: the free energy is expressed in terms of interpolating parameters $(t, \vec{x})$, meant as time and space variables, and it is shown to fulfil a transport-like equations; an explicit expression for the original free energy can then be obtained by solving a partial differential equations and suitably setting the parameters $(t, \vec{x})$.
Exploiting this technique and assuming replica-symmetry (RS), we recover the result previously found in \cite{ZiqquratBarra} for multi-layer spin-glasses and, further, our approach allows us to obtain a refined picture which includes replica-symmetry breaking; only the first step (1RSB) is addressed in details, the generalisation works analogously \cite{ottaviani}.
\newline
Next, we move to a more complex architecture, where
we introduce an additional class of spins, one for each couple of adjacent layers and bridging the related spins, see Fig.~ \ref{fig:schema2} (left panel).
The interest in this kind of structure lays in the fact that, as we prove, it is formally equivalent to a modular Hopfield model, also referred to as deep Hopfield network (DHN) see Fig.~\ref{fig:schema2} (right panel).
Indeed, in the past years, following experimental evidences about the existence of sub-units in brain networks (see e.g., \cite{bullmore, kumar, zhao, moretti}) associative neural networks embedded in modular topologies have attracted much attention \cite{dyson,AMT-JSP2019,tanaka,ozawa,happel}.
Here, exploiting the above-mentioned interpolating techniques, we solve for the free energy of a Hopfield model made of several sub-units, each made of fully-connected neurons. We first address the problem under the replica-symmetry assumption and then we extend the treatment to the first-step of replica symmetry breaking.

\section{The transport equation for the DBM}

In this section we focus on the MSK model as a formal representation of the DBM: first, in sec.~\ref{sec:def_MSK}, we formally introduce the model and the related observables, next, in sec.~\ref{sec:mech_MSK} we prove that its interpolating quenched pressure fulfils a transport-like equation and finally, in secs.~\ref{sec:RS_MSK}-\ref{sec:RSB_MSK} we obtain a solution of such an equation -- under, respectively, the RS and the 1RSB assumption -- and, by suitably setting the interpolating parameters, we get an explicit expression for the MSK quenched pressure in the thermodynamic limit.

\begin{figure}
\centering
\includegraphics[width=0.4\textwidth]{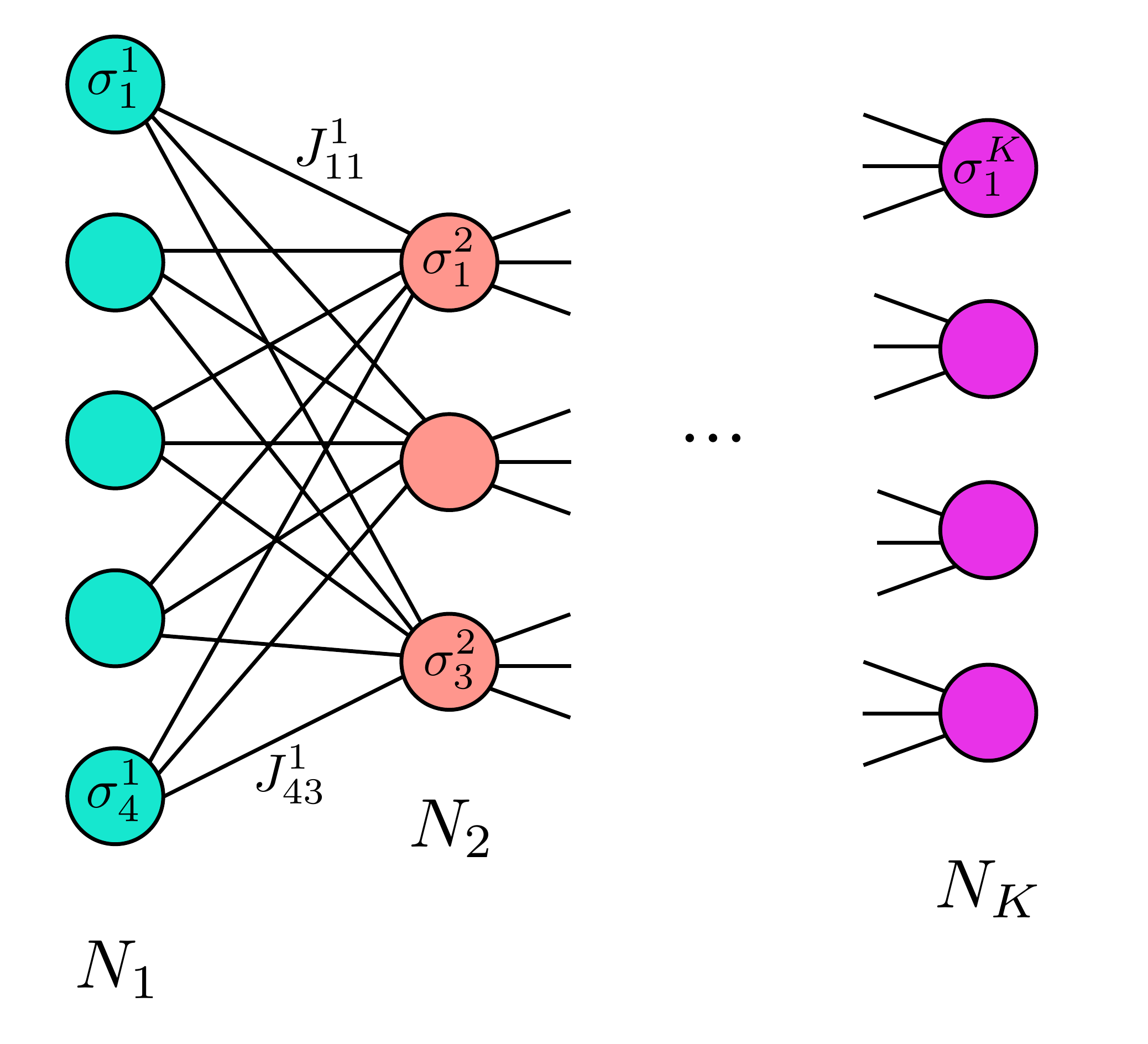}
\caption{Schematic representation of a restricted DBM, made of $K$ layers of size $N_1=5$, $N_2=3$, ..., $N_K=4$. Spins belonging to adjacent layers are completely connected pairwisely. Notation is specified only for a few spins and links in order to ensure readability.} \label{fig:schema1}
\end{figure}

\begin{figure}
\centering
\includegraphics[width=0.45\textwidth]{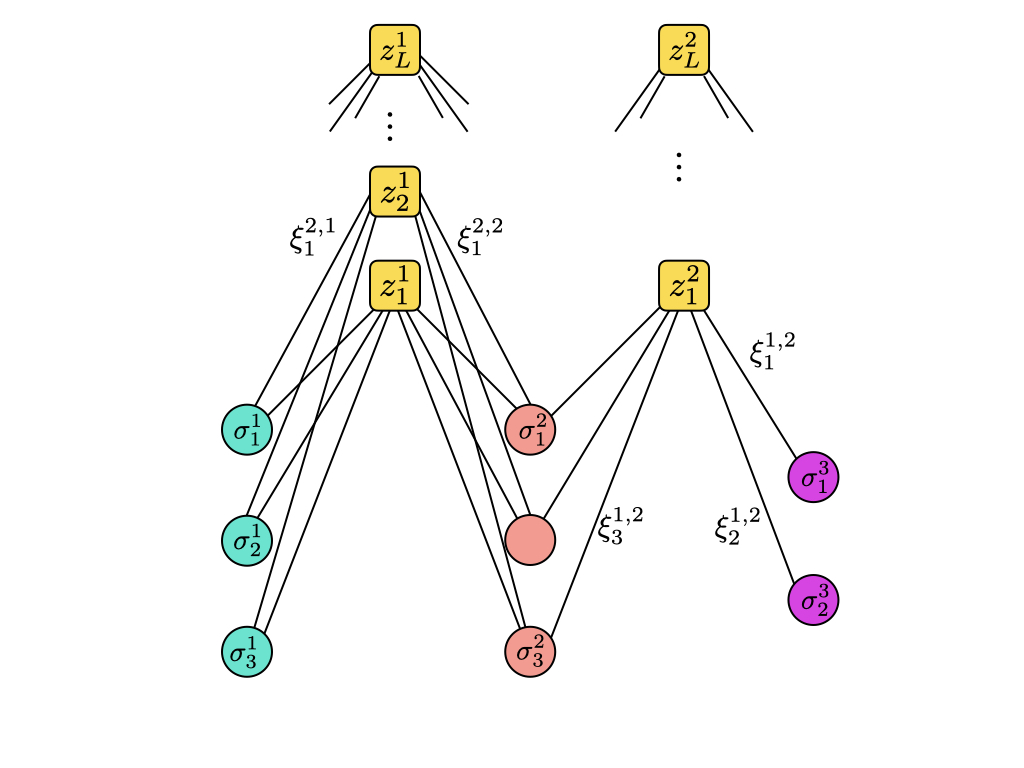}
\includegraphics[width=0.5\textwidth]{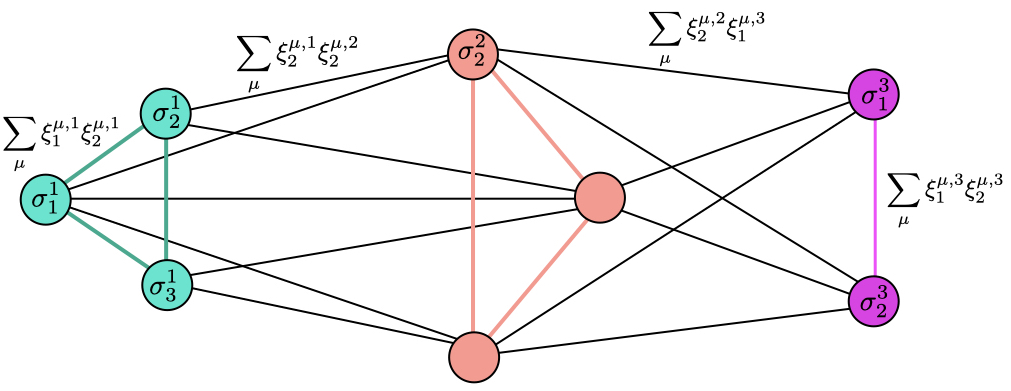}
\caption{Left: Schematic representation of a DBM, made of $K=3$ layers (the depth is limited for readability purposes), made of $N_1 = N_2 =3$ and $N_3=2$ spins; spins belonging to adjacent layers, say layers $p$ and $p+1$ are connected to the hidden neuron $z_p$, thus, the overall number of hidden neurons is $K-1$. Right: Schematic representation of a DHM, corresponding to the DBM on the left; this is made of $K=3$ layers of size $N_1 = N_2 =3$ and $N_3=2$. Each layer includes spins that are fully connected each other, further, spins belonging to adjacent layers are completely connected each other. In both panels notation for spins and weights is reported only in few cases to ensure readability.} \label{fig:schema2}
\end{figure}

\subsection{Definitions} \label{sec:def_MSK}

\begin{Definition} \label{def:H_MSK}
Let us consider a multilayer Sherrington-Kirkpatrick (MSK) model of size $N$ and endowed with $K$ layers, each made of $N_p$ binary spins, with $N_p = \lambda_p N$, for $p=1,...,K$, then $\sum_{p=1}^K \lambda_p = 1$; interactions are pairwise and only involve spins belonging to adjacent layers, in such a way that the Hamiltonian of the model reads as
\beq
H_{N,K}(\sigma ; J) := -\sqrt{\frac{2}{N}}\sum_{p=1}^{K-1}\sum_{i,j=1}^{N_{p},N_{p+1}}J_{ij}^{p}\sigma^p_i\sigma^{p+1}_j,
\eeq
where $\boldsymbol \sigma^{p} = (\sigma_1^p, ..., \sigma_{N_p}^p) \in \{ -1, +1\}^{N_p}$ for $p=1,...,K$ and the matrix $\boldsymbol J^p$ has size $N_p \times N_{p+1}$ with entries that are standard Gaussian, that is $J^p_{ij} \sim \calN(0,1)$ for $i=1,...,N_p$, $j=1,...,N_{p+1}$, $p=1,...,K-1$.
\label{dsg_hbare}
\end{Definition}
\begin{Definition}The partition function of the MSK model defined in (\ref{def:H_MSK}) is given by
\beq \label{eq:Z_MSK}
Z_{N,K}(\beta,J) := \sum_{\boldsymbol \sigma} e^{-\beta H_{N,K}(\sigma ; J)}=\sum_{\boldsymbol \sigma} e^{ \beta\sqrt{\frac{2}{N}}\sum_{p=1}^{K-1}\sum_{i,j=1}^{N_{p},N_{p+1}}J_{ij}^{p}\sigma^p_i\sigma^{p+1}_j},
\eeq
where $\beta$ is the inverse temperature and the sum runs over all spin configurations, that is $\boldsymbol \sigma = ( \boldsymbol \sigma^1, \boldsymbol \sigma^2, ..., \boldsymbol \sigma^K ) \in \{-1, +1\}^N$.
\label{dsg_BareZ}
\end{Definition}
\begin{Definition}The quenched, intensive pressure of the MSK model defined in (\ref{def:H_MSK}) with partition function given by (\ref{eq:Z_MSK}) is defined as
\beq
A_{N,K}(\beta) := \frac{1}{N} \bbE_{J} \log Z_{N,K}(\beta,J),
\label{dsg_BareA}
\eeq
where $\bbE_J$ stands for the quenched averaging operator acting as $\bbE_J f(\mathbf{J}) := \int \prod_{p=1}^{K+1}\left(\frac{dJ^p}{\sqrt{2\pi}}e^{-(J^p)^2/2}\right) ~f(\mathbf{J})$ $\mathbf J=(\boldsymbol J^1,\boldsymbol{J}^{2},\dots , \boldsymbol J^{K-1})$.
We recall that the pressure equals the free energy, a constant $-\beta$ apart.
\end{Definition}
\begin{Definition} For the generic observable $O(\boldsymbol {\sigma})$, the Boltzmann-Gibbs average stemming from (\ref{eq:Z_MSK}) is defined as
\beq
\langle O \rangle: = \frac{1}{Z_{N,K}(\beta,J)} \sum_{\sigma } O(\boldsymbol {\sigma}) e^{-\beta H_{N,K}(\sigma; J)}.
\label{BG_average}
\eeq
In relation to the Boltzmann-Gibbs average we also introduce the $\Delta$ operator for the observable $O(\boldsymbol \sigma )$ as
\beq
\Delta[ O(\boldsymbol {\sigma})]: = O(\boldsymbol {\sigma}) - \lim_{N\to \infty} \bbE_J \mv{O}.
\eeq

\end{Definition}
In the following, we will often move to the thermodynamic limit $N \to \infty$, still retaining $K$ finite. In order to highlight that a quantity is evaluated in the thermodynamic limit, we will drop the dependence on $N$ and, in particular,
\beq
A_{K}(\beta): = \lim_{N \to \infty} A_{N,K}(\beta).
\eeq

\subsection{Mechanical Analogy} \label{sec:mech_MSK}
In this subsection we introduce an interpolating pressure $\mathcal{A}_{N,K}(t,\vec x)$ depending on the interpolating parameters $t \in \mathbb{R}^+$ and $\vec{x} \in \mathbb{R}^K$, which can be interpreted as, respectively, time and space variables, and such that $A_{N,K}(\beta) = \mathcal{A}_{N,K}(t = \beta^2,\vec x =0)$; next, we will show that $\mathcal{A}_{N,K}(t,\vec x)$ fulfills a transport equation, whose solution, evaluated in $t=\beta^2$ and $\vec x=0$, therefore provides the pressure for the original MSK.

\begin{Definition}The interpolating pressure $\mathcal{A}_{N,K}(t,\vec x)$ for the MSK model (\ref{def:H_MSK}), also referred to as Guerra Action, is defined as
\bes
\mathcal{A}_{N,K}(t,\vec x) :=\frac{1}{N} \bbE_{J} \log \sum_{\bm\sigma }\exp \Big( \sqrt{t}\sqrt{\frac{2}{N}}\sum_{p=1}^{K-1}\sum_{i,j=1}^{N_{p},N_{p+1}}J_{ij}^{p}\sigma^p_i\sigma^{p+1}_j+\sum_{p=1}^K\sqrt{x_p}\sum_{i=1}^{N_p} J^p_i \sigma^p_i\Big),
\label{dsg_GuerraAction}
\ees
where $t \in \mathbb{R}^+$, $\vec{x} = (x_1, ..., x_K) \in \mathbb{R}^K$, and $J^p = (J_1^p, ..., J_{N_p}^p)$, with $J_i^p \sim \mathcal{N}(0,1)$ for  $i=1,...,N_p$ and $p=1,...,K$. Here, $\bbE_J$ stands for the averaging operator acting on both the original $J_{ij}^{p}$ couplings and the $J_i^p$ auxiliary random variables.
\end{Definition}

\begin{Remark}
	Notice that, for $t=\beta^2$ and $\vec x=0$, we exactly recover the pressure \eqref{dsg_BareA} for the original model.
\end{Remark}
\begin{Definition} Given two replicas of the system, labelled as $a$ and $b$, and characterized by the same realization of quenched disorder $\{\boldsymbol J^p\}_{p=1,...,K-1}$, we define the two-replica overlap related to the $p$-th layer as
\beq
q^p_{ab}(\boldsymbol{\sigma}^{p,(a)},\boldsymbol{\sigma}^{p,(b)}) = \frac{1}{N_p}\sum_{i=1}^{N_p} \sigma_i^{p,(a)}\sigma_i^{p,(b)},\quad \quad p=1, ..., K,
\label{dsg_orderparameters}
\eeq
and these quantities play as order parameters for the MSK model.
\end{Definition}
\begin{Lemma} The partial derivatives of the Guerra Action w.r.t. $t$ and each $x_p$ yield to the following expectation values:
\bea
\label{dsg_expvalsa}
\frac{\partial \mathcal A_{N,K}}{\partial t} &=&\sum_{p=1}^{K-1} \lambda_p\lambda_{p+1} (1- \bbE_J\mv{q^p_{12}q^{p+1}_{12}}),\\
\label{dsg_expvalsb}
\frac{\partial \mathcal A_{N,K}}{\partial x_p}  &=& \frac{\lambda_p}{2}  \big[1- \bbE_J\mv{q^p_{12}}\big],
\eea
where now, with a little abuse of notation, the thermodynamic average $\langle \cdot \rangle$ refers to the interpolating system.
\begin{proof}
The explicit computation of the derivative $\partial_t A_{N,K}$ leads to
\begin{equation}
  \label{dsg_pd_proof1}
\partial_t \mathcal{A}_{N,K} =\frac{1}{N}   \frac{1}{2\sqrt{t}}\sqrt{\frac{2}{N}}\sum_{p=1}^{K-1}\sum_{i,j=1}^{N_{p},N_{p+1}} \bbE_{J}  J_{ij}^{p} \langle \sigma^p_i\sigma^{p+1}_j \rangle,
\end{equation}
thus, applying Wick's Theorem for Gaussian averages, we obtain
\bes
  \label{dsg_pd_proof2}
  \partial_t \mathcal{A}_{N,K} &=\frac{1}{N}   \frac{1}{2\sqrt{t}}\sqrt{\frac{2}{N}}\sum_{p=1}^{K-1}\sum_{i,j=1}^{N_{p},N_{p+1}} \bbE_{J} \partial_{ J_{ij}^{p}} \langle \sigma^p_i\sigma^{p+1}_j \rangle
  = \frac{1}{N^2}\sum_{p=1}^{K-1}\sum_{i,j=1}^{N_{p},N_{p+1}} \bbE_{J} \big(1 - \langle \sigma^p_i\sigma^{p+1}_j \rangle^2 \big).
\ees
Upon using \eqref{dsg_orderparameters}, we directly obtain \eqref{dsg_expvalsa}. The computation of the spatial derivative \eqref{dsg_expvalsb} follows the same lines, thus we omit the proof.
\end{proof}
\end{Lemma}
\begin{Remark}We introduce the variables $\lambda_0=\lambda_{K+1} = 0$ in order to present the following equations in a more symmetrical fashion.
\end{Remark}
\begin{Proposition} At finite $N$, the Guerra Action \eqref{dsg_GuerraAction} obeys to the following PDE:
\beq
\pder{\mathcal A_{N,K}}{t}-2\sum_{p=1}^K (\lambda_{p-1}q^{p-1}+\lambda_{p+1}q^{p+1})\pder{\mathcal A_{N,K}}{x_p} = S_N(t,\vec x) + V_N(t, \vec x),
\label{dsg_PDE}
\eeq
where we defined
\beq
q^p :=\lim_{N\to \infty} \bbE_J \mv{q^p_{12}}, \quad \quad p=1,...,K,
\eeq
and
\begin{eqnarray}
 \label{eq:S}
 S_N(t,\vec x)&: = &\sum_{p=1}^K \lambda_{p}\lambda_{p+1}(1-q^p-q^{p+1}+q^pq^{p+1}),\\
 \label{eq:potential1}
 V_N(t,\vec x) &:= &-\sum_{p=1}^K {\lambda_p\lambda_{p+1}}{}\bbE_J \mv{\Delta [q^p_{12}]\, \Delta [q^{p+1}_{12}]},
\end{eqnarray}
are respectively the ``source'' and ``potential'' terms.

\begin{proof}
The proof works by direct use of equations (\ref{dsg_expvalsa}) and (\ref{dsg_expvalsb}). First of all, we express the correlation function of the two overlaps in terms of the averages of the fluctuations, {\it i.e.} adopting the decomposition
\begin{equation}
\mv{q^p_{12}q^{p+1}_{12}} = -q^p q^{p+1} + q^p\mv{q^{p+1}_{12}} + q^{p+1}\mv{q^p_{12}} + \mv{\Delta [q^p_{12} ]\Delta [q^{p+1}_{12}]}.
\end{equation}
Inserting this relation in \eqref{dsg_expvalsa}, we reach
\begin{equation}
\notag
\frac{\partial \mathcal A_{N,K}}{\partial t} =\sum_{p=1}^{K-1} \lambda_p\lambda_{p+1} -  \bbE_J \sum_{p=1}^{K-1} \lambda_p\lambda_{p+1} \big(-q^p q^{p+1} + q^p\mv{q^{p+1}_{12}} + q^{p+1}\mv{q^p_{12}} + \mv{\Delta [q^p_{12}] \Delta [q^{p+1}_{12}]}\big).
\end{equation}
The single-overlap expectation values ({\it i.e.} $\langle q_{12}^p\rangle$ and $\langle q_{12}^{p+1}\rangle$) can be expressed in terms of the spatial derivatives of the Guerra Action by reverting \eqref{dsg_expvalsb}. The remaining contributions, involving expectation values of the fluctuations and a polynomial function of the parameters $q^p$, are respectively condensed in source and the potential contribution. Thus, with simple algebra we reach the thesis \eqref{dsg_PDE}.
\end{proof}
\end{Proposition}

\begin{Remark}
The solution of the PDE \eqref{dsg_PDE} can be addressed by means of the method of characteristics. This suggests to consider the spatial coordinated $\vec x$ as a function of $t$ such that $\dot{x}_p:= d x_p/dt = - 2(\lambda_{p-1}q^{p-1} + \lambda_{p+1}q^{p+1})$, for $p=1,...,K$. With this choice, the PDE in Eq.~\eqref{dsg_PDE} can be recast into an evolutive ODE as
\begin{equation}
\frac{d \mathcal{A}_{N,K}}{dt} = S_N(t,\vec x) + V_N(t, \vec x).
\end{equation}
\end{Remark}

\subsection{Replica Symmetric Solution} \label{sec:RS_MSK}
In this section, we find an explicit solution for $A_K(\beta)$ under the replica symmetry (RS) assumption.
\begin{Definition} \label{def:HM_RS}
Under the replica-symmetry assumption, the distribution of the two-replica overlap $q^p_{12}$ converges in the thermodynamic limit to a Dirac delta centered at the equilibrium value $\bar{q}^p$, that is
\begin{equation}
\lim_{N \rightarrow + \infty} P_N(q_{12}^p) =  \delta (q_{12}^p - \bar{q}^p), ~~ \textrm{for}~ p=1,...,K. \label{RSdef}
\end{equation}
\end{Definition}
\begin{Proposition} In the thermodynamic limit and under the RS assumption, the Guerra Action is given by
\bes
\mathcal A^{\RS}_K (t, \vec x)  =\log 2 &+\sum_{p=1}^K \lambda_p \bbE_J \log \cosh\left[J \sqrt{x_p +2t(\lambda_{p-1}\bar q^{p-1}+\lambda_{p+1} \bar q^{p+1})}\right]   \\
&+t\sum_{p=1}^K \lambda_{p}\lambda_{p+1}(1- \bar q^p-\bar q^{p+1}+\bar q^p \bar q^{p+1}).
\label{dsg_GArs}
\ees
\begin{proof}
Under the RS assumption, the fluctuations of the order parameters w.r.t. to their equilibrium value vanish in the thermodynamic limit, meaning that ${\Delta [q_{12}^p]}\overset{N\to\infty}{\longrightarrow}0$, for any $p$. Thus, recalling the definition (\ref{eq:potential1}), we have $V_N \overset{N\to\infty}{\longrightarrow} 0$. The PDE \eqref{dsg_PDE} therefore becomes the transport equation
\beq
\pder{\mathcal A^{\RS}_K }{t}-2\sum_{p=1}^K (\lambda_{p-1}\bar q^{p-1}+\lambda_{p+1} \bar q^{p+1})\pder{\mathcal A^{\RS}}{x_p} = S^{\RS}(t,\vec x).
\label{dsg_PDEtl}
\eeq
where now $S^{\RS}(t,\vec x)=\sum_{p=1}^K \lambda_{p}\lambda_{p+1}(1-\bar q^p-\bar q^{p+1}+ \bar q^p \bar q^{p+1})$ because of the RS assumption \eqref{RSdef}.
%
This equation can be readily solved by means of the method of characteristics, from which we find that
\beq
\mathcal A_K^{\RS} \big(t,\vec x\big) = \mathcal A_K^{\RS} \big(0,\vec x_0\big)  +t S^{\RS}(t,\vec x),
\eeq
where $\vec{x}_0:= \vec{x}(t=0)$ with generic $p$-th component
\bes
\label{eq:x0}
x_{0,p}=x_p +2t(\lambda_{p-1} \bar q^{p-1}+\lambda_{p+1} \bar q^{p+1}).
\ees
The initial condition $\mathcal A_K^\RS \big(0,\vec x_0\big)$ can be explicitly obtained by using the definition \eqref{dsg_GuerraAction}, since at $t=0$ the two-body contribution disappears. In this case, the computation of the Guerra Action is straightforward (since it reduces to a 1-body model), and leads to
\bes
\mathcal A_{N,K}(0,\vec x_0) =\frac{1}{N} \bbE_{J} \log \sum_{\bm\sigma}\exp \Big ( \sum_{p=1}^K\sqrt{x_{0,p}}\sum_{i=1}^{N_p} J^p_i \sigma^p_i\Big)
=\log 2 +\sum_{p=1}^K \lambda_p \bbE_J \log \cosh(J \sqrt{x_{0,p}}).
\label{dsg_GAinit}
\ees
We stress that, once the sum over all possible configurations is taken, we end up with the product of averages of identically distributed functions of the auxiliary variables $J_i^p$, so we can drop both the spin and layer indices and use a single random variable $J$ on which the average $\bbE_{J}$ is performed. Putting the two pieces together and using \eqref{eq:x0}, we obtain the thesis \eqref{dsg_GArs}.
\end{proof}
\end{Proposition}
\begin{Theorem} The replica-symmetric intensive quenched pressure in the thermodynamic limit for the MSK model is recovered \cite{ZiqquratBarra}:
\bes
A_K^{\RS}\big(\beta)=\log 2 &+\sum_{p=1}^K \lambda_p \bbE_J \log \cosh(J \beta \sqrt{2}\sqrt{\lambda_{p-1} \bar q^{p-1}+\lambda_{p+1} \bar q^{p+1}})  + 
\beta^2\sum_{p=1}^K \lambda_{p}\lambda_{p+1}(1- \bar q^p)(1- \bar q^{p+1}).
\label{dsg_IPrs}
\ees
\begin{proof}
The result follows by simply setting $t=\beta^2$ and $\vec x=0$ in \eqref{dsg_GArs}.
\end{proof}
\end{Theorem}

\begin{Corollary} The replica-symmetric expectation of the order parameters $\bar q_{12}^p$ obeys the following self-consistent equation:
\beq
\label{eq:dsg_sc}
\bar q^p=\bbE_J \tanh^2 \big[ \beta J \sqrt{2(\lambda_{p-1} \bar q^{p-1}+\lambda_{p+1} \bar q^{p+1})} \big],\quad \quad p=1,..., K.
\eeq
\begin{proof}
To prove this result, as usual in disorder statistical mechanics, we extremize $A_K^{\RS}(\beta)$ w.r.t all order parameters: $\partial A_K ^{\RS}(\beta )/\partial \bar q^p=0$ for all $p=1,\dots,K$.
With straightforward computations, we find that extremality conditions implies
\bes
&\lambda_p \lambda_{p+1} (\bar q^{p+1}-\bbE_{J} \tanh^2[\beta J\sqrt{2(\lambda_p q^p +\lambda_{p+2}q^ {p+2})}])\\+
&\lambda_p \lambda_{p-1} (\bar q^{p-1}-\bbE_{J} \tanh^2[\beta J\sqrt{2(\lambda_p q^p +\lambda_{p-2}q^ {p-2})}])=0.
\ees
Since these constraints hold for general $\lambda_p$, we must set to zero the quantities in round brackets, which directly leads to the thesis \eqref{eq:dsg_sc}.
\end{proof}
\end{Corollary}

\subsection{1RSB solution}  \label{sec:RSB_MSK}
In this section we find an explicit solution for $A_K(\beta)$ under the 1RSB assumption which can be stated as

\begin{Definition} \label{def:HM_RSB}
In the first step of replica-symmetry breaking, the distribution of the two-replica overlap $q_{12}^p$ in the thermodynamic limit displays two delta-peaks at the equilibrium values, referred to as $\bar{q}_1^p,\ \bar{q}_2^p$. The concentration of the system at the equilibrium on these values is ruled by $\theta \in [0,1]$, that we assume to be independent on $p$:
\begin{equation}
\lim_{N \rightarrow + \infty} P_N(q_{12}^p) = \theta \delta (q_{12}^p - \bar{q}_1^p) + (1-\theta) \delta (q_{12}^p - \bar{q}_2^p), ~~~ \textrm{for} ~~ p=1,...,K. \label{limforq2}
\end{equation}
\end{Definition}

\begin{Definition}
Given the interpolating parameters $t$, $\vec x = (x_1^{(1)}, \hdots, x_K^{(1)}, x_1^{(2)}, \hdots , x_K^{(2)})$, and the i.i.d. auxiliary fields $\{ J_i^{1,(1)}, \hdots , J_i^{K,(1)}, J_i^{1,(2)}, \hdots, J_i^{K,(2)},\}_{i=1,...,N}$ with $J_i^{p,(1,2)} \sim \mathcal N (0,1)$ for $i=1, ...,N$ and $p=1,...,K$, we can write the 1-RSB interpolating partition function $\mathcal Z_{N,K}(t, \vec x)$ recursively, starting by
\begin{align}
\label{eq:Z2}
\mathcal Z_2 (t, \vec x)&=  \sums \exp \Big( \sqrt{t} \sqrt{\frac{2}{N}} \sum_{p=1}^{K-1} \sum_{ij=1}^{N_p, N_p+1} J_{ij}^p \si^p \sj^{p+1} + \suma \sum_{p=1}^K \sqrt{x_p^{(a)}} \sum_{i=1}^{N_p} J_i^{p, (a)} \si^p \Big),
\end{align}
and then averaging out the fields one per time in the following way:
\begin{align}
\mathcal Z_1(t, \vec x) \coloneqq & \ \mathbb E_2 \big [ \mathcal Z_2(t, \vec x)^\theta \big ]^{1/\theta}, \\
\mathcal Z_0(t, \vec x) \coloneqq &\ \exp \mathbb E_1 \big[ \log \mathcal Z_1(t, \vec x) \big ],\\
\mathcal Z_{N,K}(t, \vec x) \coloneqq & \ \mathcal Z_0(t, \vec x).
\end{align}
Here, $\mathbb E_2$ and $\mathbb E_1$ denote the average over the variables $J_i^{p,(2)}$'s and $J_i^{p,(1)}$'s, respectively, while $\mathbb E_0 \equiv \mathbb{E}_J$ stands for the average over the variables $J_{ij}^p$'s.
\end{Definition}

\begin{Definition}
The 1RSB interpolating pressure at finite volume $N$ is defined as
\begin{equation}\label{AdiSK1RSB}
\mathcal A^{\RSB}_{N,K} (t, \vec x) \coloneqq \frac{1}{N} \mathbb E_0 \left [ \log \mathcal Z_{N,K}(t, \vec x ) \right],
\end{equation}
and, in the thermodynamic limit,
\begin{equation}
\mathcal A_K^{\RSB} (t, \vec x) \coloneqq \lim_{N \to \infty} \mathcal A_{N,K}^{\RSB} (t, \vec x ).
\end{equation}

\end{Definition}

\begin{Remark}
Again, setting $t=\beta^2, \vec x=  0$, the interpolating pressure we recover the standard pressure (\ref{dsg_BareA}), that is, $A_N^{\RSB}(\beta, J) = \mathcal A_{N,K}^{\RSB} (t =\beta^2, \vec x =\vec 0)$.
\end{Remark}

\begin{Remark}
In order to lighten the notation, we introduce the weight
\begin{equation}
\mathcal W_2 := \frac{\mathcal Z_2^\theta}{ \mathbb E_2 \left [\mathcal Z_2^\theta \right ]},
\end{equation}
and, for each $p=1, 	\hdots, K$, we define
\begin{align}
\label{eq:sg_orders1}
 \langle q_{12}^p \rangle_1 \coloneqq & \ \mathbb{E}_0 \mathbb{E}_1 \Big[\frac{1}{N_p}\sum_{i=1}^{N_p} \left( \mathbb E_2 \big[\mathcal W_2\langle\sigma_i^p \rangle\big] \right)^2 \Big], \\
\langle q_{12}^p \rangle_2 \coloneqq &\ \mathbb{E}_0 \mathbb{E}_1 \mathbb{E}_2 \Big [\mathcal W_2\frac{1}{N_p}\sum_{i=1}^{N_p} \langle\sigma_i^p\rangle^2 \Big], \\
 \langle q_{12}^p q_{12}^{p+1} \rangle_1 \coloneqq &\  \mathbb{E}_0 \mathbb{E}_1 \Big[\frac{1}{N_p N_{p+1}}\sum_{i,j=1}^{N_p, N_{p+1}} \Big( \mathbb E_2 \big[\mathcal W_2\langle\sigma_i^p \sigma_j^{p+1}\rangle\big] \Big)^2 \Big], \\
 \label{eq:sg_orders4}
\langle q_{12}^p q_{12}^{p+1} \rangle_2 \coloneqq &\ \mathbb{E}_0 \mathbb{E}_1 \mathbb{E}_2 \Big [\mathcal W_2\frac{1}{N_p N_{p+1}}\sum_{i, j=1}^{N_p, N_{p+1}} \langle\sigma_i^p \sigma_j^{p+1}\rangle^2 \Big].
\end{align}
We also define $Q_{1,2}^p:=\lambda_p \bar{q}_{1,2}^p$ in order to lighten the notation.
\end{Remark}

\begin{Lemma} \label{lemma:2}
The partial derivatives of the interpolating quenched pressure read as
\begin{align}
\label{eq:2app}
\dt \mathcal A_{N,K}^{\RSB} &= \sum_{p=1}^{K-1} \lambda_p \lambda_{p+1} \left[ 1-(1-\theta)\langle q_{12}^p q_{12}^{p+1} \rangle_2 - \theta \langle q_{12}^p q_{12}^{p+1} \rangle_1\right] ,\\
\frac{\partial}{\partial x_p^{(1)}} \mathcal A_{N,K}^{\RSB} &= \frac{\lambda_p}{2} \left[1-(1-\theta)\langle q_{12}^p \rangle_2- \theta \langle q_{12}^p \rangle_1\right], \, &\textnormal{for} \ p=1, \hdots K \label{eq:x1app},\\
\frac{\partial}{\partial x_p^{(2)}} \mathcal A_{N,K}^{\RSB} &= \frac{\lambda_p}{2}\left[1-(1-\theta)\langle q_{12}^p \rangle_2\right], \,  &\textnormal{for} \ p=1, \hdots K  ,\label{eq:x2app}
\end{align}
\end{Lemma}

\begin{proof}
The proof of the Lemma is straightforward but pretty lengthy, so we will only prove Eq. (\ref{eq:2app}). The derivation of the other equalities follows the same lines. In order to lighten the notation, we also introduce
\begin{align}
B(\bm \sigma^p; J) := \exp \Big(\sqrt{\frac{2t}{N}} \sum_{p=1}^{K-1} \sum_{i,j=1}^{N_p, N_p+1} J_{ij}^p \si^p \sj^{p+1} + \suma \sum_{p=1}^K \sqrt{x_p^{(a)}} \sum_{i=1}^{N_p} J_i^{p, (a)} \si^p  \Big).
\end{align}
We have
\begin{align}
\dt \mathcal A_{N,K}^{\RSB} =& \frac{1}{N} \mathbb{E}_0 \mathbb{E}_1 \mathbb{E}_2 \Big[ \mathcal{W}_2 \frac{1}{\mathcal Z_2} \sums B(\bm \sigma^p; J) \frac{1}{2}\sqrt{\frac{2}{tN}}\sum_{p=1}^{K-1} \sum_{i,j=1}^{N_p, N_{p+1}} J_{ij}^p \si^p \sj^{p+1}\Big] = \notag \\
=& \frac{1}{2N\sqrt{tN}} \mathbb{E}_0 \mathbb{E}_1 \mathbb{E}_2 \Big[\sum_{p=1}^{K-1} \sum_{i,j=1}^{N_p, N_{p+1}} \partial_{J_{ij}^p} \Big( \mathcal{W}_2 \frac{1}{\mathcal Z_2}\sums B(\bm \sigma^p; J) \si^p \sj^{p+1}\Big)\Big],
\end{align}
where in the last line we used Wick's Theorem. The computation of the derivative w.r.t $J_{ij}^p$ can be decomposed as
\begin{align}
&\partial_{J_{ij}^p} \Big( \mathcal{W}_2 \frac{1}{\mathcal Z_2}\sums B(\bm \sigma^p; J) \si^p \sj^{p+1}\Big) = B_1+B_2+B_3,
\end{align}
where the first contributions is
\begin{align}
&B_1 := \left( \partial_{J_{ij}^p}  \mathcal W_2 \right) \frac{1}{\mathcal{Z}_2} \sums B(\bm \sigma^p; J) \si^p \sj^{p+1} = \sqrt{\frac{2t}{N}} \Big[ \theta \mathcal{W}_2 \Big(\frac{1}{Z_2} \sums B(\bm \sigma^p; J) \si^p \sj^{p+1}\Big)^2  \notag \\
& \ \ \ \ \ \  -\theta \mathcal{W}_2 \frac{1}{\mathcal{Z}_2} \sums B(\bm \sigma^p; J) \si^p \sj^{p+1}\mathbb{E}_2 \Big( \mathcal{W}_2 \frac{1}{\mathcal{Z}_2} \sums B(\bm \sigma^p; J) \si^p \sj^{p+1}\Big) \Big] ,
\end{align}
while the other two terms are respectively
\begin{align}
&B_2 := \mathcal{W}_2  \Big(\partial_{J_{ij}^p}\frac{1}{\mathcal{Z}_2}\Big)\sums B(\bm \sigma^p; J) \si^p \sj^{p+1}= - \sqrt{\frac{2t}{N}}\mathcal{W}_2 \Big(\frac{1}{\mathcal{Z}_2} \sums B(\bm \sigma^p; J) \si^p \sj^{p+1} \Big)^2 ,\\
&B_3 :=\mathcal{W}_2 \frac{1}{\mathcal{Z}_2}\Big(\partial_{J_{ij}^p}  \sums B(\bm \sigma^p; J) \si^p \sj^{p+1} \Big) = \sqrt{\frac{2t}{N}}\mathcal{W}_2 \frac{1}{\mathcal{Z}_2}  \sums B(\bm \sigma^p; J) (\si^p \sj^{p+1})^2 .
\end{align}
Reassembling all the terms and using the definitions \eqref{eq:sg_orders1}-\eqref{eq:sg_orders4}, we obtain the thesis \eqref{eq:2app}.
\end{proof}

\begin{Proposition}
\label{prop:3}
The streaming of the 1-RSB interpolating quenched pressure obeys, at finite volume $N$, a standard transport equation, that reads as 
\begin{align}
\label{eq:stream1rsb}
\frac{d \mathcal A_{N,K} }{dt}&=\frac{\partial\mathcal A_{N,K}}{\partial t} +\sum_{p=1}^K\Big( \dot x_p^{(1)} \frac{\partial\mathcal A_{N,K}}{\partial x^{(1)}}  +\dot x_p^{(2)} \frac{\partial \mathcal A_{N,K}}{\partial x^{(2)}}  \Big)= S_N(t, \vec x) + V_N(t, \vec x ),
\end{align}
where now
\begin{align}
\label{S}
S_N(t, \vec x) &\coloneqq \sumpnew \lambda_p \lambda_{p+1} +(1-\theta)Q_2^p Q_2^{p+1}  + \theta Q_1^p Q_1^{p+1} +\sump \lambda_p \big(Q_2^{p-1} + Q_2^{p+1}\big),\\
\label{V}
V_N(t, \vec x) &\coloneqq - \sumpnew \lambda_p \lambda_{p+1} \big[  (1-\theta) \langle \Delta q_{12}^p \Delta q_{12}^{p+1} \rangle_2 + \theta \langle \Delta q_{12}^p \Delta q_{12}^{p+1} \rangle_1\big].
\end{align}
\end{Proposition}

\begin{proof}
The proof follows the same lines as the RS case. Indeed, we start by expressing the correlation function of the overlaps in terms of their fluctuations and the spatial derivatives of the interpolating quenched pressure, ultimately grouping the remaining terms in the source and potential contributions. Starting from the derivatives w.r.t. $t$, we have
\begin{align}
&\frac{\partial\mathcal A_{N,K}}{\partial t} = \sum_{p=1}^{K-1}\Big[ \lambda_p \lambda_{p+1} -\lambda_p \lambda_{p+1}(1-\theta) \langle \Delta q_{12}^p \Delta q_{12}^{p+1} \rangle_2 +(1-\theta) Q_2^{p}Q_2^{p+1} - (1-\theta) Q_2^p \lambda_{p+1}\langle q_{12}^{p+1} \rangle_2  \notag \\
&-(1-\theta)Q_2^{p+1} \lambda_p \langle q_{12}^{p} \rangle_2  - \lambda_p \lambda_{p+1} \theta \langle \Delta q_{12}^p \Delta q_{12}^{p+1} \rangle_1 + \theta Q_1^{p}Q_1^{p+1} - \theta Q_1^p \lambda_{p+1} \langle q_{12}^{p+1} \rangle_1 -\theta Q_1^{p+1} \lambda_p \langle q_{12}^{p} \rangle_1 \Big] = \notag \\
&= \sum_{p=1}^{K-1} \Big[-\lambda_p \lambda_{p+1} [  (1-\theta) \langle \Delta q_{12}^p \Delta q_{12}^{p+1} \rangle_2 + \theta \langle \Delta q_{12}^p \Delta q_{12}^{p+1} \rangle_1] +  \lambda_p \lambda_{p+1}+(1-\theta)Q_2^p Q_2^{p+1} - \theta Q_1^p Q_1^{p+1} \notag \\
&- \lambda_p (Q_2^{p-1} + Q_2^{p+1}) +2(Q_1^{p-1} + Q_1^{p+1})\frac{\partial\mathcal A_{N,K}}{\partial x_p^{(1)}}   + 2(Q_2^{p-1}-Q_1^{p-1} + Q_2^{p+1}-Q_1^{p+1})\frac{\partial\mathcal  A_{N,K}}{\partial x_p^{(2)}} \Big].
\end{align}
\normalsize
Thus, posing
\begin{align}
\dot x_p^{(1)} &= 2(Q_1^{p-1} + Q_1^{p+1}),  ~~ &\textrm{for} ~~ p=1, \hdots , K,\\
\dot x_p^{(2)}&=2(Q_2^{p-1}-Q_1^{p-1} + Q_2^{p+1}-Q_1^{p+1}), ~~ &\textrm{for} ~~ p=1, \hdots , K,
\end{align}
we reach the thesis.
\end{proof}



\begin{Proposition}
The transport equation associated to the interpolating pressure of the MSK model (\ref{def:H_MSK}), in the thermodynamic limit and in the 1RSB scenario, reads as
\begin{align}
\frac{\partial \mathcal A^{\RSB}_{K}}{\partial t} - &\sum_{p=1}^K \left[ 2(Q_1^{p-1} + Q_1^{p+1}) \frac{\partial  \mathcal A^{\RSB}_{K}}{\partial x^{(1)}} + 2(Q_2^{p-1}-Q_1^{p-1} +Q_2^{p+1}-Q_1^{p+1}) \frac{\partial\mathcal A^{\RSB}_{K}}{\partial x^{(2)}}  \right]= \notag \\
&=\sum_{p=1}^K [\lambda_p \lambda_{p+1} +(1-\theta)Q_2^p Q_2^{p+1}  + \theta Q_1^p Q_1^{p+1} + \lambda_p (Q_2^{p-1} + Q_2^{p+1})],
\label{1rsbtranseqprima}
\end{align}
whose solution is given by
\begin{align}
\mathcal A_K^{\RSB}(t, \vec x) & =\log2 +\sump  \frac{\lambda_p}{\theta} \mathbb{E}_1 \Big[\log \mathbb{E}_2 \Big( \cosh^\theta \big( \suma \sqrt{x_{0,p}^{(a)}} J^{p, (a)} \big)\Big)\Big]\notag \\
&+t\sump [\lambda_p \lambda_{p+1}+(1-\theta)Q_2^p Q_2^{p+1}  + \theta Q_1^p Q_1^{p+1} + \lambda_p (Q_2^{p-1} + Q_2^{p+1})].
\label{1rsbtranseq}
\end{align}
\end{Proposition}

\begin{proof}
The procedure is analogous to the RS case, with the only difference that here we have to consider fluctuations around two possible equilibrium values. First, we notice that, in the thermodynamic limit and in the 1-RSB scenario under investigation, we have for all $p=1, \hdots , K$
\begin{align}
\lim_{N \rightarrow + \infty} \langle q_{12}^p \rangle_1 &= \qb_1^p ,\label{limforq1RSB} \\
\lim_{N \rightarrow + \infty} \langle q_{12}^p \rangle_2 &= \qb_2^p, \label{limforq21RSB}
\end{align}
in such a way that the potential in (\ref{eq:stream1rsb}) is vanishing, that is
\begin{equation} \label{eq:V0RSB}
\lim_{N \to \infty} V_N^{\RSB}(t, \vec x)=0.
\end{equation}
Therefore, the PDE \eqref{eq:stream1rsb} reduces to the simpler transport equation (\ref{1rsbtranseqprima}) in the thermodynamic limit upon using (\ref{eq:V0RSB}).
This equation can again be solved via the method of the characteristics: the solution can be written in the form
\begin{align}
\mathcal A_K^{\RSB} (t, \vec x)= \mathcal A_K^{\RSB}(0, \vec x_0) + t S_N^{\RSB}(t, \vec x ),
\label{eq:A1RSB}
\end{align}
with the characteristics being
\begin{align}
x_p^{(1)} &=x_0^{(1)} - 2(Q_1^{p-1} + Q_1^{p+1}) t, \label{charx1}\\
x_p^{(2)} &= x_0^{(2)} - 2(Q_2^{p-1}-Q_1^{p-1} + Q_2^{p+1}-Q_1^{p+1})   t. \label{charx2}
\end{align}
The Cauchy condition at $t=0$ and $\vec {x}_0 := \vec x (t=0)$ can be calculated directly from (\ref{eq:Z2})-(\ref{AdiSK1RSB}), as it is again a one-body calculation. The result is
\begin{align}
\mathcal A_{N,K}^{\RSB}(0, \vec x_0) &=\sum_{p=1}^K \frac{1}{N\theta} \mathbb{E}_0  \mathbb{E}_1 \Big[ \log \mathbb{E}_2 \Big(\prod_{i=1}^{N_p}\sums \exp \big( \suma \sqrt{x_{0,p}^{(a)}}J_i^{p,(a)} \si^p \big)\Big)^\theta\ \Big]= \notag \\
&=\log 2 + \sump \mathbb{E}_1 \Big[\frac{\lambda_p}{\theta} \log \Big( \mathbb{E}_2   \cosh^{\theta}\big( \suma \sqrt{x_{0,p}^{(a)}}J^{p,(a)}\big) \Big) \Big], \label{eq:A01RSB}
\end{align}
where we also used $\sum_{p=1}^K \lambda_p =1$. Combining the two terms in \eqref{eq:A1RSB}, we reach the thesis \eqref{1rsbtranseq}.
\end{proof}

\begin{Theorem}
The 1-RSB quenched pressure for the MSK model (\ref{def:H_MSK}) in the thermodynamic limit reads as
\begin{align}
\mathcal A_K^{\RSB}(\beta,J)&=  \log 2 +\sump \lambda_p \mathbb{E}_1 \left \{ \frac{\lambda_p}{\theta} \log \Big[ \mathbb{E}_2 \cosh^{\theta}\big(\beta \sqrt{2( \lambda_{p-1}\qb_1^{p-1} + \lambda_{p+1} \qb_1^{p+1})}J^{p,(1)}\right. \notag \\
& \left.  +\beta\sqrt{2\big[\lambda_{p-1}( \qb_2^{p-1}-\qb_1^{p-1}) +\lambda_{p+1} (\qb_2^{p+1}-\qb_1^{p+1})\big]}J^{p,(2)}\big) \Big] \right \} +\beta^2\sum_{p=1}^K \big[\lambda_p \lambda_{p+1}\notag\\
\label{eq:fina}
&+ (1-\theta)\lambda_{p}\lambda_{p+1}\qb_2^p \qb_2^{p+1}  + \theta \lambda_{p} \lambda_{p+1} \qb_1^p \qb_1^{p+1} + \lambda_p (\lambda_{p-1} \qb_2^{p-1} +   \lambda_{p+1} \qb_2^{p+1})\big].
\end{align}
\end{Theorem}
\begin{proof}
For the proof, it is sufficient to set $t=\beta^2$ and $\vec x^{(1)} = \vec x^{(2)} =0$ in \eqref{1rsbtranseq}.
\end{proof}

\begin{Corollary} \label{cor:SC_SK}
The self-consistency equations for the expectation of the order parameters of the MSK model (\ref{def:H_MSK}) read as
\begin{align}
\label{seq:a}
\qb_1^p = \mathbb{E}_1 &\left \{ \frac{\mathbb{E}_2 \left[ \cosh^\theta \left(g(\mathbf   J) \right)\tanh \left(g(\mathbf   J) \right) \right]}{\mathbb{E}_2 \left [ \cosh^\theta \left(g(\mathbf   J) \right) \right ]}  \right \}^2, ~~~\textrm{for} ~~ p=1, \hdots , K, \\
\label{seq:b}
\qb_2^p = \mathbb{E}_1 &\left \{ \frac{\mathbb{E}_2 \left[ \cosh^\theta \left(g(\mathbf   J) \right)\tanh^2 \left(g(\mathbf   J) \right) \right ]}{\mathbb{E}_2 \left[ \cosh^\theta \left(g(\mathbf   J) \right) \right] }\right \},  ~~~\textrm{for} ~~ p=1, \hdots , K,
\end{align}
where
$\mathbf   J= \left(J^{1,(1)}, \hdots , J^{K, (1)},  J^{1, (2)}, \hdots , J^{K, (2)}\right)$ and $$g(\mathbf   J)= \beta \sqrt{2(\lambda_{p-1}\qb_1^{p-1} + \lambda_{p+1}\qb_1^{p+1})}J^{p,(1)}+ \beta\sqrt{2(\lambda_{p-1}\qb_2^{p-1}-\lambda_{p-1}\qb_1^{p-1} + \lambda_{p+1}\qb_2^{p+1}-\lambda_{p+1}\qb_1^{p+1})}J^{p,(2)}.$$
\end{Corollary}

\begin{proof}
First, let us resume the derivatives (\ref{eq:x1app})-(\ref{eq:x2app}) separately for each $p$ and set them in the 1RSB framework:
\begin{align}
\frac{\partial\mathcal A_K^{\RSB}  }{\partial x_p^{(1)}}&=\frac{\lambda_p}{2} - \frac{\lambda_p}{2}(1-\theta)\qb^p_2 - \frac{\lambda}{2} \theta \qb^p_1 \label{selfx1}, \\
\frac{\partial\mathcal A_K^{\RSB} }{\partial x_p^{(2)}} &= \frac{\lambda_p}{2} - \frac{\lambda_p}{2}(1-\theta)\qb^p_2 \label{selfx2}.
\end{align}
This set of equations is interpreted as a system of two equations and two unknowns $(\qb^p_1, \qb^p_2)$.
Next, we evaluate the derivatives of $A_K^{\RSB}$ w.r.t. $x_p^{(1,2)}$ starting from \eqref{eq:fina}, we plug the resulting expressions into \eqref{selfx1}-\eqref{selfx2} and, finally, with some algebra, we get \eqref{seq:a}-\eqref{seq:b}.

\end{proof}

\section{The transport equation for the DHN}
In this section we focus on the DHN: first, in Sec.~\ref{sec:DHN_def}, we introduce the model and the related observables; next, in Sec.~\ref{sec:DHN_mech} we prove that its interpolating quenched pressure fulfills a transport-like equations; finally, in Secs.~\ref{sec:DHN_RS} and \ref{sec:DHN_1RSB} we obtain a solution of such an equation -- under, respectively, the RS and the 1RSB assumption, thus obtaining explicit expressions for the DHN quenched pressure in the thermodynamic limit.

\subsection{Definitions} \label{sec:DHN_def}
\begin{Definition} \label{def:H_DHN}
Let us consider a deep Hopfield network of size $N$ and endowed with $K+1$ layers (or modules), each made of $N_p$ binary neurons, with $N_p = \lambda_p N$, for $p=1,...,K+1$ (thus $\sum_{p=1}^{K+1} \lambda_p=1$); interactions are pairwise and involve neurons belonging to the same layer or neurons belonging to adjacent layers, in such a way that the Hamiltonian of the model reads as
\beq \label{eq:hamilH}
H_{N,K}(\boldsymbol \sigma ; \boldsymbol \xi) = -\frac{1}{2N}\sum_{p=1}^K \sum_{\mu=1}^{L} \Big(\sum_{i=1}^{N_{p+1}} {\xi}_i^{\mu,p+1} \sigma_i^{p+1}+\sum_{i=1}^{N_{p}} \xi_i^{\mu,p} \sigma_i^{p}\Big)^2,
\eeq
where $\boldsymbol \sigma^p = (\sigma_1^p,...,\sigma_{N_p}^p) \in \{-1, +1 \}^{N_p}$ for $p=1,...,K+1$ and the $N_p$ entries $\{\xi_i^{\mu,p}\}_{i=1,...,N_p}$ of the vector $\boldsymbol{\xi}^{\mu,p}$ are Rademacher random variables, for $\mu=1,...,L$ and $p=1,...,K+1$.
\label{dhn_hbare}
\end{Definition}
\begin{Remark}
Expanding the square appearing in \eqref{eq:hamilH}, we see that the intra-layer interactions follow the Hebbian rule, that is
\beq
J_{i,j}^{\intra,p}:= \frac{1}{N} \sum_{\mu=1}^L \xi_i^{\mu,p}\xi_j^{\mu,p},
\eeq
while inter-layer interactions correspond to
\beq
J_{i,j}^{{\inter},p}:= \frac{2}{N} \sum_{\mu=1}^L \xi_i^{\mu,p}\xi_j^{\mu,p+1}.
\eeq
\end{Remark}
\begin{Definition}The partition function of the DHN defined in (\ref{def:H_DHN}) is given by
  \beq
  \label{dhn_BareZ}
\mathcal Z_{N,K}(\beta,\xi) := \sum_{\boldsymbol \sigma} e^{-\beta H_{N,K}(\boldsymbol \sigma; \boldsymbol \xi)}=\sum_{\sigma} \exp \Big[ \frac{\beta}{2N}\sum_{p=1}^K \sum_{\mu=1}^{L} \Big(\sum_{i=1}^{N_{p+1}} {\xi}_i^{\mu,p+1} \sigma_i^{p+1}+\sum_{i=1}^{N_{p}} \xi_i^{\mu,p} \sigma_i^{p}\Big)^2 \Big],
\eeq
where $\beta$ is the inverse temperature and the sum runs over all neuron configurations, that is $\boldsymbol \sigma = (\boldsymbol \sigma^1, \boldsymbol \sigma^2,$ $...,\boldsymbol \sigma^{K+1}) \in \{-1, +1 \}^N$.

\end{Definition}

\begin{Remark}
As common in the Machine Retrieval scenario, we will work under the single-pattern condensation assumption, that is, the relevant information to be retrieved is encoded in a single vector $\boldsymbol \xi ^{1,p}$, with all the others contributing as a slow noise source (in a standard signal-to-noise analysis). As a consequence, the vectors $\boldsymbol \xi^{1,p}$ and $\boldsymbol \xi^{\mu,p}$ for $\mu \ge 2$ are treated separately in the following analysis. By introducing $K \times (L-1)$ Gaussian random variables $z_{\mu}^p \sim \calN(0,1)$ for $\mu=2,...,L$ and $p=1,...,K$, we can rewrite \eqref{dhn_BareZ} as
\bes \label{eq:Z_HNN}
\mathcal Z_{N,K}(\beta,\xi) = \sum_{\boldsymbol \sigma}\int \calD z \, \exp \Big[  &\frac{\beta}{2N}\sum_{p=1}^K  \Big(\sum_{i=1}^{N_{p+1}}{\xi}_i^{1,p+1} \sigma_i^{p+1}+\sum_{i=1}^{N_{p}} \xi_i^{1,p} \sigma_i^{p}\Big)^2  \\
&+ \sqrt{\frac{\beta}{N}}\sum_{p=1}^K \sum_{\mu>1}^{L} \Big(\sum_{i=1}^{N_{p+1}} {\xi}_i^{\mu,p+1} \sigma_i^{p+1}+\sum_{i=1}^{N_{p}} \xi_i^{\mu,p} \sigma_i^{p}\Big) z_\mu^p \Big],
\ees
where $ \calD z:= \prod_{\mu,p} e^{-(z_{\mu}^p)^2/2}/\sqrt{2 \pi}$, and the first term inside the exponential will act as a signal condensate while the second term will carry quenched slow noise. In the following, we exploit the universality of slow noise \cite{Genovese,Agliari-Barattolo} and, in the second contribution, we replace the Boolean fields $\xi$ with Gaussian fields carrying the same lowest order statistics, namely the average and the squared average of $\xi_{i}^{\mu,p}$ are, respectively, $0$ and $1$, for any $i \in (1,N_p)$, $\mu \in (2,L)$ and $p \in (1,K)$.
\end{Remark}
\begin{Remark}
The expression in \eqref{eq:Z_HNN} can be looked at as the partition function of a multi-layer spin-glass where binary spins belonging to adjacent layers are bridged via shared Gaussian spins; the latter are overall $L \times K$ split into $K$ groups and are not directly connected each other.
\end{Remark}
\begin{Definition}The quenched free energy is defined as
\beq
\mathcal A_{N,K}(\beta) = \frac{1}{N} \bbE_{\xi} \log \mathcal Z_{N,K}(\beta,\xi),
\label{dhn_BareA}
\eeq
where $\bbE_{\xi}$ is the quenched averaging operator over the non-retrieved vectors, which acts as a Gaussian average over $\boldsymbol{\xi}^p$ for $p=2,...,K$.
\end{Definition}

\begin{Definition} For the generic observable $O(\boldsymbol {\sigma})$, the Boltzmann-Gibbs average stemming from (\ref{eq:Z_HNN}) is defined as
\beq
\langle O \rangle: = \frac{1}{Z_{N,K}(\beta, \xi)} \sum_{\sigma } O(\boldsymbol {\sigma}) e^{-\beta H_{N,K}(\sigma; \xi)}.
\label{BG_average_HNN}
\eeq
In relation to the Boltzmann-Gibbs average we also introduce the $\Delta$ operator for the observable $O(\boldsymbol \sigma )$ as
\beq
\Delta[ O(\boldsymbol {\sigma})]: = O(\boldsymbol {\sigma}) - \lim_{N\to \infty} \bbE_{\xi} \mv{O},
\eeq
and, for any integer $n \in \mathbb N$,
\beq
\Delta^n [ O(\boldsymbol {\sigma})]: = \Big[O(\boldsymbol {\sigma}) - \lim_{N\to \infty} \bbE_{\xi} \mv{O}\Big]^n.
\eeq
\end{Definition}


\subsection{Mechanical Analogy}\label{sec:DHN_mech}
In this subsection, we introduce an interpolating pressure $\mathcal{A}_{N,K}(t,\vec x,\vec y,\vec w,\vec v)$ depending on the interpolating parameters $t \in \mathbb{R}^+$, $\vec{x} \in \mathbb{R}^{K+1}$, and $\vec{y},\vec{w},\vec{v} \in \mathbb{R}^{K}$, which can be interpreted as, respectively, time and space variables, and such that $A_{N,K}(\beta) = \mathcal{A}_{N,K}(t = \beta,\vec x =0, \vec y = 0, \vec w = 0, \vec v = 0)$; next, we will show that $\mathcal{A}_{N,K}(t,\vec x,\vec y,\vec w,\vec v)$ fulfills a transport equation, whose solution, evaluated in $t=\beta$ and $\vec x= \vec y = \vec w = \vec v = 0$, therefore provides the pressure for the original DHN.

\begin{Definition}The interpolating pressure $\mathcal A_{N,K}(t,\vec x,\vec y,\vec w,\vec v)$ for the DHN (\ref{def:H_DHN}), also referred to as Guerra Action, is defined as
\bes
\mathcal A_{N,K}(t,\vec x,\vec y,\vec w,\vec v) :=\frac{1}{N} \bbE_{\xi} \log \sum_{\sigma}\int \calD z \, \exp \Big[  &\frac{t}{2N}\sum_{p=1}^K  \Big(\sum_{i=1}^{N_{p+1}} {\xi}_i^{1,p+1} \sigma_i^{p+1}+\sum_{i=1}^{N_{p}} \xi_i^{1,p} \sigma_i^{p}\Big)^2  \\
&+ \sqrt{\frac{t}{N}}\sum_{p=1}^K \sum_{\mu=2}^{L} \Big(\sum_{i=1}^{N_{p+1}} {\xi}_i^{\mu,p+1} \sigma_i^{p+1}+\sum_{i=1}^{N_{p}} \xi_i^{\mu,p} \sigma_i^{p}\Big) z_\mu^p \\
&+\sum_{p=1}^K  \Big(\sum_{i=1}^{N_{p+1}} {\xi}_i^{1,p+1} \sigma_i^{p+1}+\sum_{i=1}^{N_{p}} \xi_i^{1,p} \sigma_i^{p}\Big) v_p  \\
&+ \sum_{p=1}^{K+1} \sqrt{x_p} \sum_{i=1}^{N_p} J_i^{p} \sigma_i^p+ \sum_{p=1}^{K} \sqrt{y_p} \sum_{\mu=2}^{L} \tilde{J}_\mu^{p} z_\mu^p\\
&+\frac{1}{2} \sum_{p=1}^{K} w_p \sum_{\mu=2}^{L}( z_\mu^p)^2\Big],
\label{dhn_GuerraAction}
\ees
where $t \in \mathbb R^+$, $\vec x = (x_1, ...,x_{K+1}) \in \mathbb{R}^{K+1}$, $\vec y = (y_1, ...,y_{K}) \in \mathbb{R}^K$,  $\vec w = (w_1, ...,w_{K}) \in \mathbb{R}^K$,  $\vec v = (v_1, ...,v_{K}) \in \mathbb{R}^K$, also, $J^p \in \mathbb{R}^{N_p}$ for $p=1,...,K+1$ and $\tilde{J}^p \in \mathbb{R}^{L-1}$ for $p=1,...,K$, and the entries of $J^p$ and $\tilde{J}^p$ are i.i.d. standard Gaussian variables. Here, $\bbE_{\xi}$ stands for the average w.r.t. the non-retrieved vectors $\boldsymbol \xi^{\mu,p}$ for $\mu\ge 2$ and the auxiliary random variables $J_i^p$ and $\tilde J^p_\mu$.
\end{Definition}

\begin{Remark}
	As standard in the interpolating scenario, for $t=\beta$ and $\vec x,\vec y,\vec w,\vec v=0$ we exactly recover the pressure \eqref{dhn_BareA} for the original model.
\end{Remark}

In the following, with a little abuse of notation, in order to make equations more symmetric and valid in general, we denote $L_p=L$ and $\alpha_p =L_p/N$ for each $p=1,\dots,K $.

\begin{Definition} Given two replicas of the system labelled as $a$ and $b$, and characterized by the same realization of quenched disorder $\{\boldsymbol \xi^p\}_{p=1,...,K-1}$, we define the two-replica overlaps related to the $p$-th layer and the Mattis magnetazion related to the $\mu$-th pattern as
\bes\label{dhn_orderparameters}
q^p_{ab}(\sigma^{p,(a)},\sigma^{p,(b)}) &:= \frac{1}{N_p}\sum_{i=1}^{N_p} \sigma_i^{p,(a)}\sigma_i^{p,(b)},\quad \quad \quad \quad  p=1, ..., K+1,\\
p^p_{ab}(z^{p,(a)}, z^{p,(b)}) &:= \frac{1}{L_p}\sum_{\mu=2}^{L_p} z_\mu^{p,(a)}z_\mu^{p,(b)}, \quad \quad \quad \quad  p=1 ,..., K,\\
m_{\mu}^p(\sigma) &:=\frac{1}{N_p}\sum_{i=1}^{N_p}\xi^{\mu,p}_i\sigma_i^{p},\quad\quad \quad \quad \quad p=1, ..., K,
\ees
all of them playing the role of order parameters for the model.
\end{Definition}
\begin{Lemma} Taking partial derivatives of the Guerra Action w.r.t. $t$ and each $x_p$ produces the following expectation values:
\bea
\label{dhn_expvalsa}
\frac{\partial \mathcal A_{N,K}}{\partial t} &=& \frac{1}{2}\bbE_\xi \sum_{p=1}^K\mv{ (\lambda_{p+1}m_{p+1}^1 +\lambda_{p}m_p^1 )^2}+ \frac{1}{2}\bbE_\xi \sum_{p=1}^K \alpha_p(\lambda_p+ \lambda_{p+1})\mv{p^p_{11}}\\
														&-& \frac{1}{2}\bbE_\xi \sum_{p=1}^K \alpha_p (\lambda_p \mv{p^p_{12}q^p_{12}}+ \lambda_{p+1} \mv{p^p_{12} q^{p+1}_{12}}),\nonumber\\
\label{dhn_expvalsb}
\frac{\partial \mathcal A_{N,K}}{\partial x_p}  &=& \frac{\lambda_p}{2} \bbE_\xi [1- \mv{q^p_{12}}],\\
\frac{\partial \mathcal A_{N,K}}{\partial y_p}  &=& \frac{\alpha_p}{2} \bbE_\xi [ \mv{p^p_{11}}- \mv{p^p_{12}}],\\
\frac{\partial \mathcal A_{N,K}}{\partial w_p}  &=& \frac{\alpha_p}{2} \bbE_\xi \mv{p^p_{11}},\\
\frac{\partial \mathcal A_{N,K}}{\partial v_p} &=& \bbE_\xi \mv{\lambda_{p+1}m^{p+1}_1 +\lambda_{p}m^p_1 }.
\eea
\begin{proof}
  The calculations follow the same method presented for \eqref{dsg_expvalsa},\eqref{dsg_expvalsb}, here we will show how to proceed for $\partial_{x_p} \mathcal A_{N,K}$:
  \bes
\partial_{x_p}\mathcal A_{N,K}(t,\vec x,\vec y,\vec w,\vec v) =\frac{1}{2 N \sqrt{x_p}}  \sum_{p=1}^{K+1}  \sum_{i=1}^{N_p} \bbE_{\xi} J_i^{p} \left\langle \sigma_i^p \right\rangle,
\ees
where we used \eqref{BG_average_HNN} in order to express the Boltzmann average; we now apply Wick's theorem for Gaussian averages and replace $J_i^p$ with $\partial_{J_i^p}$ to reach:
\bes
\partial_{x_p}\mathcal A_{N,K}(t,\vec x,\vec y,\vec w,\vec v) &=\frac{1}{2 N \sqrt{x_p}}  \sum_{p=1}^{K+1}  \sum_{i=1}^{N_p} \bbE_{\xi} \partial_{J_i^p} \left\langle \sigma_i^p \right\rangle
=\frac{1}{2 N}  \sum_{p=1}^{K+1}  \sum_{i=1}^{N_p} \bbE_{\xi}(1- \left\langle \sigma_i^p \right\rangle^2),
\ees
where, by direct application of eq.~\eqref{dhn_orderparameters}, we reach our result. The proof for the remaining derivatives works analogously and shall be omitted.
\end{proof}
\end{Lemma}
\begin{Remark}
In the following, we will make use of the following notation:
\bes
T^{\lambda}_p (Z)=\lambda_p Z_p+\lambda_{p+1}Z_{p+1},\quad T^{ \alpha }_p (W)=\alpha_p W_p+\alpha_{p-1}W_{p-1},
\label{eq:DHNnot}
\ees
where $Z$ and $W$ are resp. $K+1$- and $K$-dimensional vector. Such a definition will be of great help in lightening the notation.
\end{Remark}
\begin{Proposition} The Guerra Action \eqref{dhn_GuerraAction}, at finite size $N$, obeys to the following PDE:\\
\bes
\frac{\partial \mathcal A_{N,K}}{\partial t} &-\sum_{p=1}^K T^\lambda _p (\bar q) \pder{\mathcal A_{N,K}}{y_p} -  \sum_{p=1}^K T^\lambda _p (1-\bar q)  \frac{\partial \mathcal A_{N,K}}{\partial w_p}- \sum_{p=1}^K T^\lambda _p (\bar m) \frac{\partial \mathcal A_{N,K}}{\partial v_p}  -\sum_{p=1}^{K+1}  T^{ \alpha }_p(\bar p)\pder{\mathcal A_{N,K}}{x_p} =\\
& = S_N(t,\vec x,\vec y,\vec w,\vec v)+V_N(t,\vec x,\vec y,\vec w,\vec v),
\label{dhn_PDE}
\ees
where all the parameters are defined as the thermodynamic equilibrium quantities:
\bes
\bar{q}^p &:=\lim_{N\to \infty} \bbE \mv{q^p_{12}}, \quad \quad p=1\cdots K+1,\\
\bar{p}^p &:=\lim_{N\to \infty} \bbE \mv{p^p_{12}}, \quad \quad p=1\cdots K,\\
\bar{m}^p &:=\lim_{N\to \infty} \bbE \mv{m_1^p}, \quad \quad p=1\cdots K.\\
\ees
As in the previous cases, we defined the source and the potential terms resp. as
\bes
S_{N}(t,\vec x,\vec y,\vec w,\vec v) &:= -\frac{1}{2}\sum_{p=1}^K \big(T^\lambda _p (\bar m) \big)^2   - \frac{1}{2}\sum_{p=1}^K \alpha_p \bar{p}^p T^\lambda _p (1-\bar q),\\
V_N(t,\vec x,\vec y,\vec w,\vec v)& :=\frac{1}{2}\bbE_\xi \sum_{p=1}^K\mv{ \Delta^2[T^\lambda _p ( m_1)]}   - \frac{1}{2}\bbE_\xi \sum_{p=1}^K \alpha_p \big(\lambda_p \mv{\Delta p^p_{12} \Delta q^p_{12}}+ \lambda_{p+1} \mv{\Delta p^p_{12} \Delta q^{p+1}_{12}}\big).
\ees
\begin{proof}
The proof works by direct use of equations (\ref{dhn_expvalsa},\ref{dhn_expvalsb}), indeed by applying the decompositions
\bea
\mv{p^p_{12}q^p_{12}} &=& -\bar{p}^p \bar{q}^p + \bar{q}^p\mv{p^p_{12}} + \bar{p}^p\mv{q^p_{12}} + \mv{\Delta p^p_{12} \Delta q^p_{12}}, \\
\mv{p^p_{12}q^{p+1}_{12}} &=& - \bar{p}^p \bar{q}^{p+1} + \bar{q}^{p+1}\mv{p^p_{12}} + \bar{p}^p\mv{q^{p+1}_{12}} + \mv{\Delta p^p_{12} \Delta q^{p+1}_{12}}, \\
\mv{(\lambda_{p+1}m_{p+1}^1 +\lambda_{p}m_p^1)^2} &=& -(\lambda_{p+1}\bar{m}_{p+1}+\lambda_{p}\bar{m}_p) ^2 +\mv{ \Delta^2[\lambda_{p+1}m_{p+1}^1 +\lambda_{p}m_p^1]}\\
&&+2 (\lambda_{p+1}\bar{m}_{p+1} +\lambda_{p}\bar{m}_p) \mv{\lambda_{p+1}m_{p+1}^1 +\lambda_{p}m_p^1}  ,\nonumber
\eea
to \eqref{dhn_expvalsa} we can immediately write every expectation value in terms of the derivatives \eqref{dhn_expvalsb} and in terms of expectation values of $\Delta$'s which represent the fluctuations of the model.
\end{proof}
\end{Proposition}

\subsection{Replica Symmetric Solution}\label{sec:DHN_RS}
In this section, we find an explicit solution for $ \mathcal A_K(\beta)$ under the RS assumption which can be stated as\\
\begin{Definition} \label{def:HM_RSHOP}
Under the replica-symmetry assumption, the joint distribution of the two-replica overlaps and magnetizations (namely $q^p_{12}$, $p^p_{12}$ and $m_{p}$ for $p \in 1,\cdots, K$), in the thermodynamic limit, is delta-peaked at the equilibrium values, that is
\begin{equation}
\lim_{N \rightarrow + \infty} P_N(\mathbf{q}_{12}, \mathbf{p}_{12}, \mathbf{m}) =  \prod_{p=1}^K\delta (q_{12}^p - \bar{q}^p)\delta (p_{12}^p - \bar{p}^p)\delta (m^p - \bar{m}^p). \label{RSdef_HNN}
\end{equation}

\end{Definition}

\begin{Proposition} In the thermodynamic limit and under the replica-symmetry assumption, the Guerra Action is given by
\bes
\mathcal A_K^{\RS}\big(t,\vec x,\vec y,\vec w,\vec v\big) &=\log 2+\sum_{p=1}^{K+1}  \lambda_p \bbE_{\xi}\log \cosh \Big[ \xi^{1,p-1}  (v_{p-1} +t T^\lambda_{p-1} (\bar m)
)
+\xi^{1,p}  (v_p +t T^\lambda_p (\bar m) 
)
\\&+\sqrt{x_p +t T^\alpha_p (\bar p)  
} J^{p} \Big]
+\sum_{p=1}^{K} \frac{1}{2} \alpha_p \Big[\frac{y_p +t T^\lambda _p (\bar q) 
}{1-w_p -t T^\lambda _p (1-\bar q)
}
- \log(1-w_p -t  T^\lambda _p (1-\bar q)
)
\Big]\\
& -\frac{t}{2}\sum_{p=1}^K  \big(T^\lambda_p (\bar m)\big)^2
- \frac{t}{2}\sum_{p=1}^K \alpha_p \bar p^p
T^\lambda _p (1-\bar q),
\label{dhn_GArs}
\ees
where we introduced the auxiliary variables $\alpha_0=\alpha_{K+1}=0$, $\lambda_0=\lambda_{K+2}=0$ and $\xi^{1,0} = \xi^{1,K+1} = 0$.
\begin{proof}
As a direct consequence of the RS assumption, the potential term vanishes: $V_N \overset{N\to\infty}{\longrightarrow} 0$. Therefore, the PDE \eqref{dhn_PDE} becomes the transport equation
\bes
\frac{\partial \mathcal A_K^{\RS}}{\partial t} &-\sum_{p=1}^K (\lambda_p  \bar q^p+ \lambda_{p+1}  \bar q^{p+1})\pder{ \mathcal A_K^{\RS}}{y_p} -  \sum_{p=1}^K [\lambda_p(1- \bar q^p) + \lambda_{p+1}(1- \bar q^{p+1}) ] \frac{\partial \mathcal A_K^{\RS} }{\partial w_p} \\
&-  \sum_{p=1}^R (\lambda_p \bar m^p+ \lambda_{p+1} \bar m^{p+1}) \frac{\partial \mathcal A_K^{\RS}}{\partial v_p}  -\sum_{p=1}^{K+1} ( \alpha_p \bar p^p+\alpha_{p-1} \bar p^{p-1})\pder{\mathcal A_K^{\RS}}{x_p} =\\
& =-\frac{1}{2}\sum_{p=1}^K (\lambda_{p+1} \bar m^{p+1} +\lambda_{p} \bar m^p )^2   - \frac{1}{2}\sum_{p=1}^K \alpha_p \bar p^p[\lambda_p(1- \bar q^p) + \lambda_{p+1}(1- \bar q^{p+1}) ].
\label{dhn_PDEtl}
\ees
This can be readily solved by the method of characteristics which gives
\beq
\mathcal A_K^{\RS}\big(t,\vec x,\vec y,\vec w,\vec v\big) = \mathcal A_K(0,\vec x_0,\vec y_0,\vec w_0,\vec v_0)   -\frac{t}{2}\sum_{p=1}^K \big(\lambda_{p+1} \bar m^{p+1} +\lambda_{p} \bar m^p \big)^2   - \frac{t}{2}\sum_{p=1}^K \alpha_p \bar p^p\big[\lambda_p(1- \bar q^p) + \lambda_{p+1}(1- \bar q^{p+1}) \big],
\eeq
where the initial condition for the spatial coordinates ($x_{0,p}:= x_p(t=0)$ for any $p$, and similarly for $\vec y$, $\vec w$, and $\vec v$), reads as
\bes
x_{0,p}&=x_p +t  \big( \alpha_p \bar p^p+\alpha_{p-1} \bar p^{p-1}\big),\\
y_{0,p}&=y_p +t   \big(\lambda_p  \bar q^p+ \lambda_{p+1}  \bar q^{p+1}\big),\\
w_{0,p}&=w_p +t \big[\lambda_p(1- \bar q^p) + \lambda_{p+1}(1-\bar q^{p+1}) \big],\\
v_{0,p}&=v_p +t   (\lambda_p \bar m^p+ \lambda_{p+1} \bar m^{p+1}),\\
\ees
while the initial condition for the interpolating pressure can be explicitly obtained by using  \eqref{dhn_GuerraAction}. Indeed, at $t=0$ we get
\bes
\mathcal A_{N, K}(0,\vec x_0,\vec y_0,\vec w_0,\vec v_0)& =\frac{1}{N} \bbE_{\xi} \log \sum_{\boldsymbol \sigma}\int \calD z \, \exp \Big[\sum_{p=1}^K  \Big(\sum_{i=1}^{N_{p+1}} {\xi}_i^{1,p+1} \sigma_i^{p+1}+\sum_{i=1}^{N_{p}} \xi_i^{1,p} \sigma_i^{p}\Big) v_{0,p}  \\
&+ \sum_{p=1}^{K+1} \sqrt{x_{0,p}} \sum_{i=1}^{N_p} J_i^{p} \sigma_i^p+ \sum_{p=1}^{K} \sqrt{y_{0,p}} \sum_{\mu=2}^{L_p} \tilde{J}_\mu^{p} z_\mu^p+\frac{1}{2} \sum_{p=1}^{K} w_{0,p} \sum_{\mu=2}^{L_p}( z_\mu^p)^2\Big]=\\
&=\bbE_{\xi}\sum_{p=1}^{K+1}  \lambda_p\log  \Big[2\cosh ( {\xi}^{1,p}  v_{0,p-1}+\xi^{1,p}  v_{0,p} + \sqrt{x_{0,p}} J^{p} )\Big]\\
&+\sum_{p=1}^{K} \frac{\alpha_p}{2}  \Big[\frac{y_{0,p}}{1-w_{0,p}} - \log(1-w_{0,p})\Big].
\label{dhn_GAinit}
\ees
Putting all pieces together and recalling the notation \eqref{eq:DHNnot}, we get the thesis.
\end{proof}
\end{Proposition}
\begin{Theorem} The replica-symmetric intensive quenched pressure for the DHN model (\ref{def:H_DHN}) is obtained
\bes
\mathcal A^{\RS}_K(\beta) &=\log 2+\sum_{p=1}^{K+1}  \lambda_p \bbE_{\xi}\log \cosh \big( \beta {\xi}^{1,p}   T^\lambda_{p-1}(\bar m) 
+  \beta\xi^{1,p}
T^\lambda_p (\bar m)
+\sqrt{\beta 
T^\alpha_p (\bar p)
}
 J^{p} \big)\\
&+\sum_{p=1}^{K} \frac{\alpha_p}{2}  \Big[\frac{\beta T^\lambda_p (\bar q)
}{1
-\beta T^\lambda_p  (1-\bar q)
}
- \log(1
-\beta T^\lambda _p (1-\bar q)
)
 \Big]
 - \frac{\beta}{2}\sum_{p=1}^K \big(
T^\lambda _p (\bar m)
\big)^2   - \frac{\beta}{2}\sum_{p=1}^K \alpha_p \bar p^p
T^\lambda _p (1-\bar q).
\label{dhn_IPrs}
\ees
\begin{proof}
Recalling that the pressure \eqref{dhn_GuerraAction} for the original model can be obtained by setting $t=\beta$ and $\vec x,\vec y,\vec w,\vec v=0$ in the interpolating quenched pressure, we just need to plug these values in the expression \eqref{dhn_GAinit} found before.
\end{proof}
\end{Theorem}

\begin{Corollary} The replica-symmetric expectation for the order parameters $q_{12}^p, p_{12}^p, m^p$ obey the following self-consistency equations:
\begin{align}
\label{eq:m1}
\bar{m}^p&= \mathbb{E} \left[ \tanh (g_p(\mathbf J))\right], \ \ \ & \  p=1, \hdots, K+1 ,\\
\label{eq:q1}
\bar{q}^p&= \mathbb{E} \left[ \tanh^2 (g_p(\mathbf J))\right] ,\ \ \ & \  p=1, \hdots, K+1, \\
\label{eq:p1}
\bar{p}^p &= \frac{\beta T^\lambda _p (\bar q)
}{1-\beta
T^\lambda _p (1-\bar q)
} ,\ \ \ & \ p=1, \hdots, K,
\end{align}
where
\begin{equation}
  \label{eq:gfield_rs_dhn}
  g_p(\mathbf J):=\Big[ \beta {\xi}^{1,p}\big(T^\lambda _{p-1}(\bar m)+T^\lambda _p (\bar m)
 \big) +\sqrt{\beta 
 T^\alpha_p (\bar p)	
 } J^{p} \Big],
 \end{equation}
and $\bbE$ stands for the averaging operator acting as $\bbE f(\mathbf{J}) := \int \prod_{p=1}^{K+1}\left(\frac{dJ^p}{\sqrt{2\pi}}e^{-(J^p)^2/2}\right) ~f(\mathbf{J})$.

\begin{proof}
	The thesis follows by imposing the extremality condition of the intensive quenched pressure in the thermodynamic limit w.r.t. to all of the order parameters. After straightforward (but simple) computations, one directly obtain the thesis.
\end{proof}
\end{Corollary}

\begin{Corollary}
In the single layer case, we recover the Amit-Gutfreund-Sompolinsky (AGS) intensive quenched pressure \cite{AGS}.

\begin{proof}
First of all, we fix $K=1$. In this case, the intensive pressure reads
\bes
\label{eq:2laDHN}
\mathcal A_K^{\RS}(\beta) &=\log 2+ (\lambda_1+\lambda_2) \bbE_{J}\log \cosh \Big[ \beta (\lambda_1 \bar{m}_1+ \lambda_{2} \bar{m}_2) +\sqrt{\beta  \alpha_1 \bar{p}^1} J \Big]\\
&+\frac{1}{2} \alpha_1 \Big[\frac{\beta \big(\lambda_1  \bar{q}^1+ \lambda_{2}  \bar{q}^{2}\big)}{1 -\beta \lambda_1 (1-\bar{q}^1) -\beta \lambda_{2}(1-\bar{q}^{2}) }- \log\big(1 -\beta \lambda_1(1-\bar{q}^1) -\beta \lambda_{2}(1-\bar{q}^{2})\big)\Big]\\
& -\frac{\beta}{2}\big(\lambda_{2}\bar{m}_2 +\lambda_{1}\bar{m}_1 \big)^2   - \frac{\beta}{2}\alpha_1 \bar{p}^1\big[\lambda_1(1-\bar{q}^1) + \lambda_{2}(1-\bar{q}^{2}) \big],
\ees
Since for $K=1$ the DHN reduces to the two-layer network, in order to re-obtain the standard Hopfield model we should identify
\bes
&\alpha_1=\alpha,\\
&\lambda_1 \bar{m}_1+\lambda_2 \bar{m}_2 = \bar m,\\
&\lambda_1  \bar{q}^1+ \lambda_{2}  \bar{q}^{2} = \bar q,\\
& \bar{p}^1 =\bar p,
\ees
from which, using $\lambda_1+\lambda_2=1$, we obtain
\bes
A_K^{\RS}(\beta) &=\log 2+ \bbE_{J}\log \cosh \Big[ \beta \bar m +\sqrt{\beta \alpha\bar p} J \Big] -\frac{\beta}{2}\bar{m}^2   - \frac{\beta}{2}\alpha \bar p(1-\bar q)\\
&+\frac{\alpha}{2} \frac{\beta \bar q}{1 -\beta (1-\bar q) } -\frac{\alpha}{2} \log\big(1 -\beta (1-\bar q)\big),
\ees
which equals the AGS free energy, thus proving the thesis.
\end{proof}
\end{Corollary}

\begin{Remark}
We can analyze the zero-temperature limit $\beta\to\infty$ of the self-consistency equations \eqref{eq:m1}-\eqref{eq:p1}. In particular, we see that
\bes
\label{eq:lim_for_m1}
\bar m^p = \int \frac{dJ_p}{\sqrt{2\pi}} \exp\Big(-\frac{J_p^2}{2}\Big) \tanh (g_p(\mathbf J)) \xrightarrow{\beta\rightarrow + \infty}\int \frac{dJ_p}{\sqrt{2\pi}} \exp\Big(-\frac{J_p^2}{2}\Big) \sign (g_p(\mathbf J)) =  \textnormal{erf}\Big( \frac{A_p}{\sqrt{2 B_p}} \Big),
\ees
\begin{equation}
\label{eq:lim_for_q}
\bar q^p = \int \frac{dJ_p}{\sqrt{2\pi}} \exp\Big(-\frac{J_p^2}{2}\Big) \tanh^2 (g_p(\mathbf J)) \xrightarrow{\beta\rightarrow + \infty}\int \frac{dJ_p}{\sqrt{2\pi}} \exp\Big(-\frac{J_p^2}{2}\Big) = 1,
\end{equation}
where explicitely $A_p = {\xi}^{1,p}(\lambda_{p-1} \bar{m}^{p-1}+ 2\lambda_{p} \bar{m}^{p}+ \lambda_{p+1} \bar{m}^{p+1})$ and $B_p = \alpha_p \bar{p}^p+\alpha_{p-1} \bar{p}^{p-1}$. Since $\bar q^p\to 1 $ as $\beta\to\infty$, we can define $C_p =\beta (1-\bar q^p)$, which satisfy the following self-consistency equation:
\bes
C_p := \beta(1-\bar{q}^p)&=\frac{\partial}{\partial {A_p}}\mathbb{E}\tanh (g_p(\mathbf J))\xrightarrow{\beta\rightarrow + \infty} \frac{\partial}{\partial {A_p}}  \textnormal{erf}\Big( \frac{A_p}{\sqrt{2 B_p}} \Big)=  \sqrt{\frac{2}{\pi B_p}}\exp\Big( - \frac{A_p^2}{2B_p} \Big).
\ees
Therefore, for $\beta \rightarrow  \infty$, exploiting the previous results we are left with a new set of conditions:
\begin{align}
&\bar{q}^p = 1 ,\\
& \bar{m}^p = \textnormal{erf} \left( \frac{{\xi}^{1,p}(\lambda_{p-1} \bar{m}_{p-1}+ 2\lambda_{p}, \bar{m}_{p}+\lambda_{p+1} \bar{m}_{p+1})}{\sqrt{2(\alpha_p \bar{p}^p+\alpha_{p-1} \bar{p}^{p-1})}}\right), \\
&\bar{p}^p = \frac{\lambda_p+ \lambda_{p+1}}{1- \lambda_p C_p - \lambda_{p+1}C_{p+1}},\\
&C_p =  \sqrt{\frac{2}{\pi B_p}}\exp\Big( - \frac{A_p^2}{2B_p} \Big).
\end{align}
\end{Remark}

\subsection{1 RSB solution} \label{sec:DHN_1RSB}
In this section, we find an explicit solution for $A_K(\beta)$ under the 1RSB assumption which can be stated as

\begin{Definition} \label{def:MHM_RSB}
In the first step of replica-symmetry breaking, the distribution of the two-replica overlap $q$, in the thermodynamic limit, displays two delta-peaks at the equilibrium values, referred to as $\bar{q}_1^p,\ \bar{q}_2^p$. The concentration on the two values is ruled by $\theta \in [0,1]$, namely for each $p=1, \hdots , K$, that is
\begin{equation}
\lim_{N \rightarrow + \infty} P'_N(q^p) = \theta \delta (q^p - \bar{q}^p_1) + (1-\theta) \delta (q^p - \bar{q}^p_2). \label{limforq2_HOP}
\end{equation}
Similarly, for the overlap $p$, denoting with $\bar{p}_1,\ \bar{p}_2$ the equilibrium values, we have
\begin{equation}
\lim_{N \rightarrow + \infty} P''_N(p^p) = \theta \delta (p^p - \bar{p}^p_1) + (1-\theta) \delta (p^p - \bar{p}^p_2). \label{limforp2_HOP}
\end{equation}
The magnetization $m^1_p$ still self-averages at $ \bar{m}_p$ in the thermodynamic limit.
\end{Definition}

\begin{Definition}
Given $\vec x = (\bm x^{(1)}, \bm x^{(2)},  \bm y^{(1)},  \bm y^{(2)}, \bm w, \bm v)\in \mathbb{R}^{2(K+1)+4K}$, $t \in \mathbb{R}^{+}$ as interpolating parameters and the i.i.d. auxiliary fields $\{J_i^{r,(1)}, J_i^{r,(2)}\}_{i=1,...,N_p,\  r=1, \hdots , K+1}$, $\{\tilde J_\mu^{r,(1)},\tilde J_\mu^{r, (2)}\}_{i=1,...,N_p, \ r=1, \hdots
, K}$ with $J_i^{(1,2)}, \tilde J_{\mu}^{(1,2)} \sim \mathcal N(0,1)$ for $\mu=2,...,L_p$, we can write the 1-RSB interpolating partition function $\mathcal Z_N(t, \vec x)$ for the multilayer Hopfield model recursively, starting by
\begin{eqnarray}
\label{eqn:Z2_HOP}
\mathcal Z_2(t, \vec x) &=&\sum_{\bm \sigma} \int D \bm \tau \exp \Big[ \frac{t}{2N} \sumpp \Big( \sum_{i=1}^{N_{p+1}} \xi_i^{1, p+1} \si^{p+1} + \sum_{i=1}^{N_{p}} \xi_i^{1,p} \si^p \Big)^2    \notag \\
&+&\sqrt{\frac{t}{N}} \sumpp \sum_{\mu=1}^{L_p} \Big( \sum_{i=1}^{N_{p+1}} \xi_i^{\mu, p+1} \si^{p+1}
+ \sum_{i=1}^{N_p} \xi_i^{\mu,p} \si^p \Big) \tau_\mu^p + \sumpp\Big( \sum_{i=1}^{N_{p+1}}  \xi_i^{1, p+1} \si^{p+1}
+ \sum_{i=1}^{N_p} \xi_i^{1,p} \si^p \Big)v_p \notag \\
&+& \suma \sum_{p=1}^{K+1} \sqrt{x_p^{(a)}} \sum_{i=1}^{N_p} J_i^{p, (a)} \si^p +  \suma \sumpp \sqrt{y_p^{(a)}} \sum_{\mu =2}^{L_p} \tilde{J}_\mu^{p, (a)} \tau_\mu^p + \frac{1}{2} \sumpp w_p \sum_{\mu=2}^{L_p} (\tau_\mu^p)^2\Big],
\end{eqnarray}
where the $\xi_i^\mu$'s are i.i.d. standard Gaussian random variables for $\mu \ge 2$.
Averaging out the fields recursively, we define
\begin{align}
\label{eqn:Z1_HOP}
\mathcal Z_1(t, \vec x) \coloneqq& \mathbb E_2 \big [ \mathcal Z_2(t, \vec x )^\theta \big]^{1/\theta}, \\
\label{eqn:Z0_HOP}
\mathcal Z_0(t, \vec x) \coloneqq&  \exp \mathbb E_1 \left[ \log \mathcal Z_1(t, \vec x) \right ] ,\\
\mathcal Z_{N, K}(t, \vec x) \coloneqq& \mathcal Z_0(t, \vec x) ,
\end{align}
where with $\mathbb E_a$ we mean the average over the variables $\tilde J_i^{p,(a)}$'s and $\tilde J_\mu^{p,(a)}$'s, for $a=1, 2$, and with $\mathbb{E}_0=\mathbb{E}_{\xi}$ we shall denote the average over the variables $\xi_i^{p,\mu}$ for $\mu \ge2$.
\end{Definition}

\begin{Definition}
The 1RSB interpolating pressure, at finite volume $N$, is introduced as
\begin{equation}\label{AdiHOP1RSB}
\mathcal A_{N, K} (t, \vec x) \coloneqq \frac{1}{N} \mathbb E_0 \left [ \log \mathcal Z_{N,K}(t, \vec x ) \right],
\end{equation}
and, in the thermodynamic limit,
\begin{equation}
\mathcal A_K (t, \vec x) \coloneqq \lim_{N \to \infty} \mathcal A_{N,K} (t, \vec x ).
\end{equation}
By setting $t=\beta, \vec x= \bm 0$, the interpolating pressure recovers the standard pressure (\ref{dsg_BareA}), that is, $\mathcal A_{N,K}(\beta, J) = \mathcal A_{N,K} (t =\beta, \vec x =\bm 0)$.
\end{Definition}

\begin{Remark}
In order to lighten the notation, hereafter we use the following
\begin{align}
\label{eq:par1}
 \langle q_{12}^p \rangle_1 \coloneqq \ &  \mathbb{E}_0 \mathbb{E}_1 \Big[\frac{1}{N_p}\sum_{i=1}^{N_p} \Big( \mathbb E_2 \big[\mathcal W_2\omega(\sigma_i^p)\big] \Big)^2\ \Big] ,\\
\langle q_{12}^p \rangle_2 \coloneqq \ & \mathbb{E}_0 \mathbb{E}_1 \mathbb{E}_2 \Big[\mathcal W_2\frac{1}{N_p}\sum_{i=1}^{N_p} \omega^2(\sigma_i^p) \Big] ,\\
\langle p_{11}^p \rangle \coloneqq\ & \mathbb{E}_0 \mathbb{E}_1 \mathbb{E}_2 \Big[\mathcal W_2\frac{1}{L_p}\sum_{\mu=2}^{L_p} \omega(\tau_\mu^p)^2 \Big],
\end{align}
\begin{align}
 \langle q_{12}^p p_{12}^{p} \rangle_1 \coloneqq\ &  \mathbb{E}_0 \mathbb{E}_1 \Big[\frac{1}{N_p L_{p}}\sum_{i=1}^{N_p} \sum_{\mu=2}^{L_p} \left( \mathbb E_2 \big[\mathcal W_2\omega(\sigma_i^p \tau_\mu^p)\big] \right)^2 \Big] ,\\
\langle q_{12}^p p_{12}^{p} \rangle_2 \coloneqq\ & \mathbb{E}_0 \mathbb{E}_1 \mathbb{E}_2 \Big[\mathcal W_2\frac{1}{N_p L_p}\sum_{i=1}^{N_p} \sum_{\mu=2}^{L_p} \omega^2(\sigma_i^p \tau_\mu^p) \Big],
\label{eq:par2}
\end{align}
where we defined the weight
\begin{equation}
\mathcal W_2 \coloneqq \frac{\mathcal Z_2^\theta}{ \mathbb E_2 \left [\mathcal Z_2^\theta \right ]}.
\end{equation}

Moreover, in order to simplify the following notation, we define $Q_a^p \coloneqq \lambda_p \qb_a^p$, $P_a^p \coloneqq \alpha_p \bar{p}_a^p$, $M^p=\lambda_p \mb^p$, for $a=1,2$ and $p=1,\hdots K$. 
\end{Remark}

\begin{Lemma} \label{lemma:2_HOP}
The partial derivatives of the interpolating quenched pressure read as
\begin{align}
\label{eq:2appHOP}
\frac{\partial\mathcal A_{N,K} }{\partial t} &= \ \frac{1}{2} \sumpp  \{ \langle (\lambda_{p+1} m_p^1 + \lambda_p m_p^1)^2 \rangle + \alpha_p \lambda_{p+1} [ \langle p_{11}^p \rangle - (1-\theta) \langle q_{12}^{p+1}p_{12}^p \rangle_2 - \theta \langle q_{12}^{p+1} p_{12}^p \rangle_1 ]  \notag \\
& +\alpha_p \lambda_p\left[ \langle p_{11}^p \rangle - (1-\theta) \langle q_{12}^{p}p_{12}^p \rangle_2 - \theta \langle q_{12}^{p} p_{12}^p \rangle_1  \right] \}, \\
\frac{\partial\mathcal A_{N,K}}{\partial x_p^{(1)}}  &= \  \frac{\lambda_p}{2} \left[1-(1-\theta)\langle q_{12}^p \rangle_2- \theta \langle q_{12}^p \rangle_1\right] ,\ \ \ \textnormal{for} \  p=1, \hdots ,K+1, \label{eq:x1appHOP}\\
\frac{\partial\mathcal A_{N,K} }{\partial x_p^{(2)}} &= \ \frac{\lambda_p}{2}\left[ 1-(1-\theta)\langle q_{12}^p\rangle_2\right] ,\ \ \ \textnormal{for} \  p=1, \hdots ,K+1 , \label{eq:x2appHOP} \\
\frac{\partial \mathcal A_{N,K}}{\partial y_p^{(1)}}& = \ \frac{\alpha_p}{2} \left[\langle p_{11}^p \rangle -(1-\theta)\langle p_{12}^p \rangle_2- \theta \langle p_{12}^p \rangle_1\right], \ \ \ \textnormal{for} \  p=1, \hdots ,K, \label{eq:y1appHOP}\\
\frac{\partial\mathcal A_{N,K}}{\partial y_p^{(2)}} & = \ \frac{\alpha_p}{2}\left[\langle p_{11}^p \rangle -(1-\theta)\langle p_{12}^p\rangle_2\right], \ \ \ \textnormal{for} \ p=1, \hdots ,K , \label{eq:y2appHOP} \\
\frac{\partial\mathcal{A}_{N,K}}{\partial v_p} &= \ \langle \lambda_{p+1} m_{p+1}^1 + \lambda_p m_p^1 \rangle ,\ \ \ \textnormal{for} \ p=1, \hdots ,K, \\
\frac{\partial \mathcal{A}_{N,K} }{\partial w_p}&= \ \frac{\alpha_p}{2} \langle p_{11}^p \rangle ,\ \ \ \ \textnormal{for} \ p=1, \hdots ,K .
\end{align}
\end{Lemma}

\begin{proof}
The proof of this Lemma is pretty lengthy, thus we will only prove (\ref{eq:2appHOP}). To this aim, we define
\begin{align}
B(\bm \sigma^p; J) &= \exp \Big[ \frac{t}{2N} \sumpp \Big( \sum_{i=1}^{N_{p+1}} \xi_i^{1, p+1} \si^{p+1} + \sum_{i=1}^{N_{p}} \xi_i^{1,p} \si^p \Big)^2   \notag \\
&+\sqrt{\frac{t}{N}} \sumpp \sum_{\mu=1}^{L_p} \Big( \sum_{i=1}^{N_{p+1}} \xi_i^{\mu, p+1} \si^{p+1}
+ \sum_{i=1}^{N_p} \xi_i^{\mu,p} \si^p \Big) \tau_\mu^p + \sumpp\Big( \sum_{i=1}^{N_{p+1}} \xi_i^{1, p+1} \si^{p+1}
+ \sum_{i=1}^{N_p} \xi_i^{1,p} \si^p \Big)v_p \notag \\
&+ \suma \sum_{p=1}^{K+1} \sqrt{x_p^{(a)}} \sum_{i=1}^{N_p} J_i^{p, (a)} \si^p +  \suma \sumpp \sqrt{y_p^{(a)}} \sum_{\mu =2}^{L_p} \tilde{J}_\mu^{p, (a)} \tau_\mu^p + \frac{1}{2} \sumpp w_p \sum_{\mu=2}^{L_p} (\tau_\mu^p)^2\Big].
\end{align}
Then, with straightforward computations and the application of the Wick's Theorem, we obtain
\begin{align}
\frac{\partial\mathcal A_{N,K} }{\partial t} 
&=\ \frac{1}{2N\sqrt{tN}} \sumpp \mathbb{E}_0 \mathbb{E}_1 \mathbb{E}_2  \Big\{ \Big( \sum_{\mu=2}^{L_p} \sum_{i=1}^{N_p+1} \partial_{{\xi}_i^{\mu, p+1}} \Big[ \mathcal W_2 \frac{1}{\mathcal{Z}_2} \sums \int D \bm \tau B(\bm \sigma^p; J) \si^{p+1} \tau_\mu^p \Big]\Big) \notag \\
&+\Big(\sum_{\mu=2}^{L_p} \sum_{i=1}^{N_p} \partial_{{\xi}_i^{\mu, p}} \Big[ \mathcal W_2 \frac{1}{\mathcal{Z}_2} \sums \int D \bm \tau B(\bm \sigma^p; J) \si^{p} \tau_\mu^p \Big]\Big) \Big\} + \frac{1}{2} \sumpp \langle (\lambda_{p+1} m_{p+1}^{\mu} + \lambda_p m_p^\mu)^2 \rangle.
\end{align}
We compute the derivative w.r.t ${\xi}_i^{\mu, p+1}$ separately; the derivative w.r.t. ${\xi}_i^{\mu, p}$ is analogous.
\begin{align}
\partial_{{\xi}_i^{\mu, p+1}} &\Big( \mathcal{W}_2 \frac{1}{\mathcal Z_2}\sums \int D \bm \tau B(\bm \sigma^p; J) \si^{p+1}\tau_\mu^p \Big) = B_1+B_2+B_3 .
\end{align}
For the first term, we have
\begin{align}
B_1 =&  \left( \partial_{{\xi}_i^{\mu, p+1}}\mathcal W_2 \right) \frac{1}{\mathcal{Z}_2} \sums \int D \bm \tau B(\bm \sigma^p; J) \si^{p+1} \tau_\mu^p = \\
=& \sqrt{\frac{t}{N}} \Big[ \theta \mathcal{W}_2 \Big(\frac{1}{\mathcal Z_2}\sums \int D \bm \tau B(\bm \sigma^p; J) \si^{p+1} \tau_\mu^p \Big)^2 - \theta \mathcal{W}_2 \frac{1}{\mathcal{Z}_2} \sums \int D \bm \tau B(\bm \sigma^p; J) \tau_\mu^p \si^{p+1}\cdot\notag \\
&\cdot\mathbb{E}_2 \Big( \mathcal{W}_2 \frac{1}{\mathcal{Z}_2} \sums \int D \bm \tau B(\bm \sigma^p; J) \si^{p+1} \tau_\mu^p \Big) \Big].
\end{align}
For the other two terms, we get
\begin{align}
B_2 =& \mathcal W_2  \left( \partial_{{\xi}_i^{\mu, p+1}}\frac{1}{\mathcal{Z}_2}  \right) \sums \int D \bm \tau B(\bm \sigma^p; J) \si^{p+1} \tau_\mu^p= - \sqrt{\frac{t}{N}}\mathcal{W}_2 \Big(\frac{1}{\mathcal{Z}_2} \sums \int D \bm \tau B(\bm \sigma^p; J) \si^{p+1} \tau_\mu^p \Big)^2, \\
B_3 =& \mathcal W_2 \frac{1}{\mathcal{Z}_2}  \Big( \partial_{{\xi}_i^{\mu, p+1}} \sums \int D \bm \tau B(\bm \sigma^p; J) \si^{p+1} \tau_\mu^p  \Big) =\sqrt{\frac{t}{N}}\mathcal{W}_2 \frac{1}{\mathcal{Z}_2}  \sums \int D \bm \tau B(\bm \sigma^p; J) (\si^{p+1} \tau_\mu^p)^2 .
\end{align}
Reassembling all the terms and recalling the Defs. \eqref{eq:par1}-\eqref{eq:par2}, we obtain the thesis.
\end{proof}

\begin{Proposition}
\label{prop:3HOP}
The streaming of the 1-RSB interpolating quenched pressure obeys, at finite volume $N$, a standard transport equation, that reads as
\begin{align}
\label{eq:stream1rsbHOP}
\frac{d \mathcal A_{N,K}}{dt}&=\frac{\partial \mathcal A_{N,K}}{\partial t}+\sum_{p=1}^{K+1} \Big(\dot x_p^{(1)} \frac{\partial\mathcal A_{N,K} }{\partial x_p^{(1)}}  +\dot x_p^{(2)} \frac{\partial\mathcal A_{N,K}}{\partial x_p^{(2)}}  \Big) + \sumpp \Big(\dot y_p^{(1)} \frac{\partial \mathcal A_{N,K} }{\partial y_p^{(1)}} +\dot y_p^{(2)} \frac{\partial \mathcal A_{N,K}}{\partial y_p^{(2)}} \Big)  \notag \\
&+\dot w_p \frac{\partial \mathcal A_{N,K}}{\partial w_p}   + \dot v_p \frac{\partial\mathcal A_{N,K} }{\partial v_p}   = S_N(t, \vec x) + V_N(t, \vec x ),
\end{align}
where
\begin{align}
\label{SHOP}
S_N(t, \vec x) &\coloneqq \sum_{p=1}^K - \frac{1}{2}\left(M^{p+1} + M^p\right)^2 + \frac{1}{2}(1-\theta)Q_2^{p+1}P_2^p \\
&\notag+ \frac{1}{2}\theta Q_1^{p+1}P_1^p +  \frac{1}{2}(1-\theta)Q_2^p P_2^p + \frac{1}{2}\theta Q_1^p P_1^p -
{ \frac{\lambda_{p+1}}{2} (P_2^{p+1} - P_2^p}), \\
\label{VHOP}
V_N(t, \vec x) &\coloneqq \sumpp \frac{1}{2}  \langle (\lambda_{p+1} (m_{p+1}^1 - \bar{m}^{p+1}) + \lambda_p (m_p^1  - \bar{m}^p))^2 \rangle -\frac{\alpha_p \lambda_{p+1}}{2}(1-\theta)\langle \Delta q_{12}^{p+1} \Delta p_{12}^p \rangle_2 \notag \\
& - \frac{1}{2}\alpha_p \lambda_{p+1} \theta \langle \Delta q_{12}^{p+1} \Delta p_{12}^p \rangle_1 - \frac{1}{2}\alpha_p \lambda_p(1-\theta)\langle \Delta q_{12}^{p} \Delta p_{12}^p \rangle_2 - \frac{1}{2}\alpha_p \lambda_p \theta \langle \Delta q_{12}^{p} \Delta p_{12}^p \rangle_1.
\end{align}
\normalsize
\end{Proposition}

\begin{proof}
	
As standard in these cases, we start by expressing the $t$-derivative of the free energy, thus expressing everything in terms of the remaining derivatives and the fluctuations of the order parameters. The latter will be taken into account in the potential $V_N$, while the remaining quantities will form the source contribution $S_N$. For the sake of clearness of presentation, we directly report the result of the $t$-derivative of the quenched free energy:
\begin{align}
\partial_t A_{N, K}& =
 \sumpp \frac{1}{2}  \langle (\lambda_{p+1} (m_{p+1}^1 - \bar{m}^{p+1}) + \lambda_p (m_p^1  - \bar{m}^p))^2 \rangle -\frac{\alpha_p \lambda_{p+1}}{2}(1-\theta)\langle \Delta q_{12}^{p+1} \Delta p_{12}^p \rangle_2 \notag \\
&- \frac{\alpha_p \lambda_{p+1}}{2} \theta \langle \Delta q_{12}^{p+1} \Delta p_{12}^p \rangle_1 - \frac{\alpha_p \lambda_p}{2}(1-\theta)\left\langle \Delta q_{12}^{p} \Delta p_{12}^p \right\rangle_2 - \frac{\alpha_p \lambda_p }{2}\theta \left\langle \Delta q_{12}^{p} \Delta p_{12}^p \right\rangle_1 \notag \\
&-\frac{1}{2}(M^{p+1} +M^p)^2 + \frac{1}{2}(1-\theta)Q_2^{p+1}P_2^p + \frac{1}{2}\theta Q_1^{p+1}P_1^p + \frac{1}{2}(1-\theta)Q_2^p P_2^p + \frac{1}{2}\theta Q_1^p P_1^p  \notag \\
&+ (M^{p+1} +M^p) \partial_{v_p} \mathcal{A}_{N, K } +(\lambda_{p+1} + \lambda_p) \partial_{w_p} \mathcal{A}_{N,K}+ (Q_1^{p+1} + Q_1^p)\partial_{y_p^{(1)}} \mathcal{A}_{N, K}  \notag \\
&+(Q_2^{p+1} - Q_1^{p+1} + Q_2^p - Q_1^p) (\partial_{y_p^{(2)}} \mathcal{A}_{N, K} - \frac{\alpha_p}{2} \langle p_{11}^p \rangle) +\sum_{p=1}^{K+1} (P_1^p + P_1^{p-1}) (\partial_{x_p^{(1)}} \mathcal{A}_{N,K} - \frac{\lambda_p}{2})  \notag \\
&+ (P_2^p - P_1^p + P_2^{p-1} - P_1^{p-1})(\partial_{x_p^{(2)}} \mathcal{A}_{N,K} - \frac{\lambda_p}{2}) -(Q_1^{p+1} + Q_1^{p}) \partial_{w_p} \mathcal{A}_{N, K} .
\end{align}
Despite the complexity of this expression, it is easy to show that, setting
\begin{align}
\dot x_p^{(1)} &= - (P_1^p + P_1^{p-1}) ,\ \ \ &\textnormal{for} \ p=1, \hdots , K+1,\\
\dot x_p^{(2)}&= - (P_2^p-P_1^p + P_2^{p-1} - P_1^{p-1}), \ \ \ &\textnormal{for} \ p=1, \hdots , K+1, \\
\dot y_p^{(1)} &= - (Q_1^p + Q_1^{p+1}), \ \ \ &\textnormal{for}\  p=1, \hdots , K,\\
\dot y_p^{(2)}&= - (Q_2^p-Q_1^p + Q_2^{p+1} - Q_1^{p+1}) ,\ \ \ &\textnormal{for}\  p=1, \hdots ,K, \\
\dot w_p&= - (\lambda_p+ \lambda_{p+1}- Q_2^p-Q_2^{p+1}), \ \ \ &\textnormal{for} \ p=1, \hdots , K, \\
\dot v_p &= - (M^{p+1} + M^P), \ \ \ &\ \textnormal{for}\  p=1, \hdots , K,
\end{align}
we get the thesis.
\end{proof}

\begin{Remark} \label{r:latterHOP}
In the thermodynamic limit and in the 1-RSB scenario, we have for all $p=1, \hdots , K$
\begin{align}
\lim_{N \rightarrow + \infty} \langle q_{12}^p \rangle_1 &= \qb_1^p ,\label{limforq1RSBHOP} \\
\lim_{N \rightarrow + \infty} \langle q_{12}^p \rangle_2 &= \qb_2^p, \label{limforq21RSBHOP} \\
\lim_{N \rightarrow + \infty} \langle p_{12}^p \rangle_1 &= \bar{p}_1^p ,\label{limforp1RSBHOP} \\
\lim_{N \rightarrow + \infty} \langle p_{12}^p \rangle_2 &= \bar{p}_2^p, \label{limforp21RSBHOP}
\end{align}
in such a way that the potential in (\ref{eq:stream1rsb}) is vanishing, that is
\begin{equation} \label{eq:V0RSBHOP}
\lim_{N \to \infty} V_N(t, \vec x)=0.
\end{equation}
\end{Remark}

Exploiting Remark \ref{r:latterHOP} we can prove the following
\begin{Proposition}
The transport equation associated to the interpolating pressure of the multilayer Hopfield model, in the thermodynamic limit and in the 1RSB scenario, reads as
\small
\begin{align}
&\frac{\partial \mathcal A^{\RSB}_{K} }{\partial t} - \sum_{p=1}^{K+1}\Big[ (P_1^p + P_1^{p-1}) \frac{\partial \mathcal A^{\RSB}_{K} }{\partial x_p^{(1)}}- (P_2^p-P_1^p + P_2^{p-1} - P_1^{p-1})    \frac{\partial\mathcal A^{\RSB}_{K} }{\partial x_p^{(2)}}\Big] -\sum_{p=1}^K \Big[(Q_1^p + Q_1^{p+1}) \frac{\partial \mathcal{A}^{\RSB}_{K} }{\partial y_p^{(1)}}  \notag \\
&+ (Q_2^p-Q_1^p + Q_2^{p+1} -  Q_1^{p+1}) \frac{\partial\mathcal{A}^{\RSB}_{K}}{\partial y_p^{(2)}}  +  (\lambda_p + \lambda_{p+1}- Q_2^p-Q_2^{p+1} ) \frac{\partial\mathcal{A}^{\RSB}_{K}}{\partial w_p}  +(M^{p+1} + M^p) \frac{\partial\mathcal{A}^{\RSB}_{K}}{\partial v_p}\Big]=\notag \\
&= \sum_{p=1}^K \Big[- \frac{1}{2}(M^{p+1} + M^p)^2 + \frac{1}{2}(1-\theta)Q_2^{p+1}P_2^p + \frac{1}{2}\theta Q_1^{p+1}P_1^p +  \frac{1}{2}(1-\theta)Q_2^pP_2^p + \frac{1}{2}\theta Q_1^p P_1^p -
 \frac{\lambda_{p+1}}{2} (P_2^{p+1} - P_2^p)\Big],
\label{1rsbtranseqprimaHOP}
\end{align}
\normalsize
whose solution is given by
\small
\begin{align}
&\mathcal A^{\RSB}_{K}(t, \vec x) = \sum_{P=1}^R \frac{\alpha_p}{2} \log \left(\frac{1}{1-w_{p,0}}\right) + \frac{\alpha_p}{2\theta} \log \left( \frac{1-w_{p,0}}{1-w_{p,0} - \theta y_{p,0}^{(2)}}\right) + \frac{\alpha_p}{2} \frac{y_{p,0}^{(1)}}{1-w_{p,0}-\theta y_{p,0}^{(2)}}   \notag \\
&+\frac{1}{\theta} \sum_{p=1}^{K+1} \mathbb{E}_1 \Big( \log \mathbb{E}_2 \big[ \lambda_p 2 \cosh ({\xi}^{1, p} v_{p-1}^0 + \xi^{1, p} v_p^0 + \sqrt{x_{p,0}^{(1)}} J^{p,(1)} + \sqrt{x_{p,0}^{(2)}} J^{p,(2)}\big]^\theta \Big)  \notag \\
&+t\Big[ \sum_{p=1}^K - \frac{1}{2}(M^{p+1} + M^p )^2 + \frac{1}{2}(1-\theta)Q_2^{p+1}P_2^p + \frac{1}{2}\theta Q_1^{p+1} P_1^p + \frac{1}{2}(1-\theta)Q_2^p P_2^p + \frac{1}{2}\theta Q_1^p P_1^p - \frac{\lambda_{p+1}}{2} ( P_2^{p+1} - P_2^p)\Big].
\label{1rsbtranseqHOP}
\end{align}
\normalsize
\end{Proposition}

\begin{proof}
The PDE (\ref{1rsbtranseqprima}) in the thermodynamic limit can be obtained from (\ref{eq:stream1rsb}) using (\ref{eq:V0RSB}).
This equation can be again solved via the method of the characteristics, with the solution given by
\begin{align}
\mathcal A_K (t, \vec x)= \mathcal A_{N, K}(0, \vec x_0) + S(t, \vec x )t.
\label{eq:A1RSBHOP}
\end{align}
In this case, the characteristics are
\begin{align}
 x_p^{(1)} &=  x_ {p, 0}^{(1)} - (P_1^p + P_1^{p-1})t, \ \ \ &\textnormal{for}\ p=1, \hdots , K+1,\\ 
x_p^{(2)}&=  x_{p,0}^{(2)} - (P_2^p-P_1^p + P_2^{p-1} - P_1^{p-1})t, \ \ \ &\textnormal{for} \ p=1, \hdots , K+1 ,\\
y_p^{(1)} &= y_{p,0}^{(1)} - (Q_1^p + Q_1^{p+1})t ,\ \ \ &\textnormal{for} \  p=1, \hdots , K,\\
y_p^{(2)}&= y_{p,0}^{(2)} - (Q_2^p-Q_1^p+ Q_2^{p+1} -  Q_1^{p+1})t, \ \ \ &\textnormal{for}\  p=1, \hdots , K, \\
w_p&= w_{p,0} - (\lambda_p + \lambda_{p+1}-Q_2^p-Q_2^{p+1})t, \ \ \ &\textnormal{for}\  p=1, \hdots , K, \\
v_p &= v_{p,0} - (M^{p+1} + M^p)t, \ \ \ &\textnormal{for}\ p=1, \hdots , K.
\end{align}
The Cauchy condition, corresponding to $t=0$ and $\vec {x_0} = \vec x (t=0)$, can be calculated directly, as it is again a one-body calculation. Its explicit expression is
\begin{align}
\mathcal A_K^{\RSB}(0, \vec x_0) &=\sum_{p=1}^K \frac{\alpha_p}{2} \log \left(\frac{1}{1-w_{p,0}}\right) + \frac{\alpha_p}{2\theta} \log \left( \frac{1-w_{p,0}}{1-w_{p,0} - \theta y_{p,0}^{(2)}}\right) + \frac{\alpha_p}{2} \frac{y_{p,0}^{(1)}}{1-w_{p,0}-\theta y_{p,0}^{(2)}}   \notag \\
&+\frac{1}{\theta} \sum_{p=1}^{K+1} \mathbb{E}_1 \big( \log \mathbb{E}_2 [ \lambda_p 2 \cosh ({\xi}^{1, p} v_{p-1}^0 + \xi^{1, p} v_p^0 + \sqrt{x_{0,p}^{(1)}} J^{p,(1)} + \sqrt{x_{0,p}^{(2)}} J^{p,(2)}]^\theta \big). \label{eq:A01RSBHOP}
\end{align}
Putting everything together, we finally reach the thesis (\ref{1rsbtranseq}).
\end{proof}

\begin{Theorem}
The 1-RSB quenched pressure for the multilayer Hopfield model, in the thermodynamic limit, reads as
\begin{align}
\mathcal A^{\RSB}_{K}(t, \vec x)& = \sum_{p=1}^K \frac{\alpha_p}{2} \log \left(\frac{1}{1-\beta T^\lambda _p(1-\bar q_2)
}\right) +   \frac{\alpha_p}{2\theta} \log \left( \frac{1-\beta
	T^\lambda_p (1-\bar q_2)
}{1-\beta
T^\lambda_p (1-\bar q_2)
 - \theta \beta 
 T^\lambda_p (\bar q_2-\bar q_1)
}\right)
 \notag \\
&+ \frac{\alpha_p}{2} \frac{\beta
	T^\lambda_p (\bar q_1)
}{1-\beta
T^\lambda _p (1-\bar q_2)
-\theta \beta 
T^\lambda _p (\bar q_2-\bar q_1)
}   \notag \\
&+\frac{1}{\theta} \sum_{p=1}^{K+1} \mathbb{E}_1 \left( \log \mathbb{E}_2 \left[ \lambda_p 2 \cosh \left({\xi}^{1, p} 
T^\lambda_{p-1}(\bar m)
 + \xi^{1, p}
 T^\lambda_{p}(\bar m)
+ \sqrt{\beta
	T^\alpha _p (\bar p_1)
} J^{p,(1)}  \notag \right. \right.\right.\\
&\left.\left. \left. + \sqrt{\beta
T^\alpha_{p-1} (\bar p_2 -\bar p_1)
} J^{p,(2)}\right)\right]^\theta \right) +\beta\Big[- \frac{1}{2} \sum_{p=1}^K 
\big(T^\lambda_p (\bar m)\big)
^2 + \frac{\alpha_p \lambda_{p+1}}{2}(1-\theta)\qb_2^{p+1}\bar{p}_2^p  \notag \\
&+ \frac{\alpha_p \lambda_{p+1}}{2}\theta \qb_1^{p+1}\bar{p}_1^p  + \frac{\alpha_p \lambda_p}{2}(1-\theta)\qb_2^p\bar{p}_2^p + \frac{\alpha_p \lambda_p}{2}\theta \qb_1^p \bar{p}_1^p-\frac{\lambda_{p+1}}{2} ( \alpha_{p+1}\bar{p}_2^{p+1} - {\alpha_p}\bar{p}_2^p) \Big].
\label{eq:finalHOP}
\end{align}
\end{Theorem}
\begin{proof}
It is sufficient to set $t=\beta$ and $\vec x =\bm 0$ in (\ref{1rsbtranseq}). 
\end{proof}

\begin{Corollary} \label{cor:SC_HOP}
The self-consistent equations for the order parameters of the multilayer Hopfield model read as
\begin{align}
\label{seq:aHOP}
&\qb_1^p= \mathbb{E}_1 \left \{ \frac{\mathbb{E}_2 \left[ \cosh^\theta \left(g_p(\mathbf J) \right)\tanh \left(g_p(\mathbf J) \right) \right]}{\mathbb{E}_2 \left [ \cosh^\theta \left(g_p(\mathbf J) \right) \right ]}  \right \}^2, \ \ \ \textnormal{for}\ p=1, \hdots , K+1, \\
\label{seq:bHOP}
&\qb_2^p = \mathbb{E}_1 \left \{ \frac{\mathbb{E}_2 \left[ \cosh^\theta \left(g_p(\mathbf J) \right)\tanh^2 \left(g_p(\mathbf J) \right) \right ]}{\mathbb{E}_2 \left[ \cosh^\theta \left(g_p(\mathbf J) \right) \right] }\right \},  \ \ \ \textnormal{for} \ p=1, \hdots , K+1, \\
&\bar{m}^p = \mathbb{E}_1 \left[ \frac{\mathbb{E}_2 \left( \cosh^\theta g _p(\mathbf J) \tanh g_p(\mathbf J)\right)}{\mathbb{E}_2 \left( \cosh^\theta g_p(\mathbf J)\right)}\right], \ \ \ \textnormal{for} \ p=1, \hdots , K ,\\
&\bar{p}_1^p= \frac{\beta
T^\lambda_p (\bar q_1)
}{(1-\beta
T^\lambda_p(1-\bar q_2)
- \theta\beta
T^\lambda _p (\bar q_2-\bar q_1)
)^2}, \ \ \ \textnormal{for} \  p=1, \hdots , K,\\
&\bar{p}_2^p = \frac{\beta
T^\lambda_p (\bar q_2 -\bar q_1)
}{(1-\beta
T^\lambda_{p+1} (1-\bar q_2)
 - \theta \beta
 T^\lambda_p (\bar q_2-\bar q_1)
)
}\frac{1}{(1-\beta
T^\lambda_p(1-\bar q_2)
)} , \ \ \ \textnormal{for} \  p=1, \hdots , K,
\end{align}
where
$\mathbf J= (J^{1,(1)}, \hdots , J^{K, (1)},  J^{1, (2)}, \hdots , J^{K, (2)})$, and
$$g_p(\mathbf J)= \beta{\xi}^{1, p} 
(T^\lambda_p (\bar m)+T^\lambda_{p-1} (\bar m))
+ \sqrt{\beta
T^\alpha _p (\bar p_1)
} J^{p,(1)} + 
\sqrt{\beta
T^\alpha _p (\bar p_2 -\bar p_1)
} J^{p,(2)}.$$
\normalsize
\end{Corollary}

\begin{proof}
Here we just sketch the proof. First, let us resume the derivatives (\ref{eq:x1app})-(\ref{eq:x2app}) separately for each $p$ and set them in the 1RSB framework
\begin{align}
\frac{\partial \mathcal A^{\RSB}_K}{\partial x_p^{(1)}} &=\frac{\lambda_p}{2} - \frac{\lambda_p}{2}(1-\theta)\qb^p_2 - \theta \qb^p_1, \label{selfx1HOP} \\
\frac{\partial  \mathcal A^{\RSB}_K}{\partial x_p^{(2)}}  &= \frac{\lambda_p}{2} - \frac{\lambda_p}{2}(1-\theta)\qb^p_2, \label{selfx2HOP} \\
\frac{\partial\mathcal A^{\RSB}_K }{\partial y_p^{(1)}}& = \frac{\partial\mathcal{A}^{\RSB}_K }{\partial w_p} - \frac{\alpha_p}{2}(1-\theta)\bar p_2^p - \theta  \bar p_{1}^p, \label{selfy1HOP} \\
\frac{\partial\mathcal A_K^{\RSB} }{\partial y_p^{(2)}}&  = \frac{\partial \mathcal{A}_K^{\RSB} }{\partial w_p} - \frac{\alpha_p}{2}(1-\theta) \bar p_{2}^p, \label{selfy2HOP} \\
\frac{\partial\mathcal{A}_K^{\RSB}}{\partial v_p}& =  \lambda_{p+1} \bar m_{p+1} + \lambda_p \bar m_p. \label{selfmHOP}
\end{align}

This set of equations is interpreted as a system of five unknowns $(\qb^p_1, \qb^p_2, \bar{p}_1^p, \bar{p}_2^p, \lambda_{p+1} \bar m_{p +1}+ \lambda_p \bar m_p )$ and five equations.
Next, we evaluate the derivatives of $A^{\RSB}_K$ starting from (\ref{eq:fina}), we plug the resulting expressions into (\ref{selfx1HOP})-( \ref{selfmHOP}) and, finally, with some algebra, we get (\ref{seq:aHOP})-(\ref{seq:bHOP}).

\end{proof}

\section{Conclusions}
In this work we considered multi-layer spin-glasses as models for deep machine-learning and showed that a rigorous statistical mechanics investigation is feasible.  

Specifically, in the first part of the paper we focused on a multi-layer Sherrington-Kirkpatrick model made of $K$ layers with $N_p$ binary spins per layer ($p=1,...,K$) interacting pairwise; couplings are allowed only between spins belonging to adjacent layers. This kind of model can be looked at as a restricted DBM with $K-2$ hidden layers (the inner ones) and $2$ visible layers (the outer ones) processing information which, in input, is codified in terms of binary vectors of size $N_1$. The statistical mechanics of this system is addressed by means of rigorous techniques based on Guerra's interpolation and we obtained an explicit expression for the related free-energy under the RS assumption and also allowing for one step of RSB. From the free energy we could also derive self-consistent equations for the two-replica overlap playing as order parameter.

In the second part of the paper, we enriched the architecture by inserting an additional set of $K-1$ layers, each made of $L$ Gaussian spins and allowing for pair-wise interactions only between binary spins and Gaussian spins belonging to adjacent layers. We showed that this kind of model displays a partition function that is equivalent to the one related to a modular HN where intra-modular as well as inter-modular interactions are permitted and these couplings provide a suitable extension of the standard Hebbian rule built over $L \times K$ binary patterns.  Again, the statistical mechanics of the system is addressed by means of Guerra's interpolation techniques and we obtained an explicit expression for the related free-energy under the RS assumption and also allowing for one step of RSB. From the free energy we could also derive self-consistent equations for the two-replica overlap and for the Mattis magnetization playing as order parameters.

\section*{Acknowledgments}\label{acknowledgments}
Sapienza University of Rome (Progetto Ateneo RM120172B8066CB0), Unisalento and INFN are acknowledged for financial support.

\end{document}